\documentclass[journal,draftclsnofoot,onecolumn]{IEEEtran}

\usepackage{amsmath}
\allowdisplaybreaks[1]
\usepackage{setspace}

\usepackage{amssymb}
\usepackage{mathrsfs}
\usepackage{amsthm}
\usepackage{multirow}
\usepackage{color}
\usepackage{longtable}
\usepackage{array}
\usepackage{url}
\usepackage{comment}
\usepackage{enumerate}
\usepackage{eucal}
\usepackage{graphics}
\usepackage{verbatim}
\usepackage{subfigure}

\usepackage{algorithm}
\usepackage{algpseudocode}
\usepackage{cite}
\usepackage[table]{xcolor}

\usepackage{mathtools}
\usepackage{xparse}
\usepackage{bm}

\DeclarePairedDelimiterX{\set}[1]{\{}{\}}{\setargs{#1}}
\NewDocumentCommand{\setargs}{>{\SplitArgument{1}{;}}m}
{\setargsaux#1}
\NewDocumentCommand{\setargsaux}{mm}
{\IfNoValueTF{#2}{#1} {#1\,\delimsize|\,\mathopen{}#2}}%{#1\:;\:#2}

\newtheorem{thm}{Theorem}

\newtheorem{lem}{Lemma}
\newtheorem{dfn}{Definition}

\newtheorem{rmk}{Remark}
\newtheorem{cor}{Corollary}

\newtheorem{example}{Example}
\newenvironment{pf}{{\noindent\it Proof:}}{\hfill $\blacksquare$\par}

\newcommand{\RNum}[1]{\lowercase\expandafter{\romannumeral #1\relax}}
\newcommand{\Rnum}[1]{\uppercase\expandafter{\romannumeral #1\relax}}

\newcommand{\F}{\mathbb{F}}

\newcommand{\eqdef}{\triangleq}
\newcommand{\cA}{\mathcal{A}}

\newcommand{\cC}{\mathcal{C}}

\newcommand{\cE}{\mathcal{E}}

\newcommand{\cN}{\mathcal{N}}
\newcommand{\cO}{\mathcal{O}}
\newcommand{\cP}{\mathcal{P}}

\newcommand{\cV}{\mathcal{V}}

\newcommand{\cH}{\mathcal{H}}

\newcommand{\ba}{\mathbf{a}}

\newcommand{\bc}{\mathbf{c}}
\newcommand{\bd}{\mathbf{d}}
\newcommand{\be}{\mathbf{e}}

\newcommand{\bt}{\mathbf{t}}
\newcommand{\bu}{\mathbf{u}}
\newcommand{\bw}{\mathbf{w}}

\newcommand{\bx}{\mathbf{x}}
\newcommand{\by}{\mathbf{y}}
\newcommand{\bz}{\mathbf{z}}

\newcommand{\Zero}{\mathbf{0}}

\newcommand{\bq}{\mathbf{q}}
\newcommand{\bh}{\mathbf{h}}

\newcommand{\GL}{{\rm GL}}
\newcommand{\In}{{\rm In}}
\newcommand{\Out}{{\rm Out}}
\newcommand{\tail}{{\rm tail}}
\newcommand{\head}{{\rm head}}

\newcommand{\Rank}{{\rm Rank}}

\newcommand{\Row}{{\rm Row}}
\newcommand{\Span}{{\rm Span}}
\newcommand{\supp}{{\rm supp}}

\title{Linear network codes for vector-linear network function computation over three-layer networks}
\author{Min Xu and  Gennian Ge%
\thanks{This research was supported by the National Key Research and Development Program of China under Grant 2025YFC3409900, the National Natural Science Foundation of China under Grant 12231014, and Beijing Scholars Program.}
  \thanks{M. Xu (e-mail: minxu0716@qq.com) is with the Institute of Mathematics and Interdisciplinary Sciences, Xidian University, Xi'an 710126, China.}%
  \thanks{G. Ge (e-mail: gnge@zju.edu.cn) is with the School of Mathematical Sciences, Capital Normal University, Beijing 100048, China.}}
\begin{document}

\maketitle
\begin{abstract}
We study vector-linear function computation over three-layer networks with a fixed target function and a fixed source-access pattern. We develop a support-constrained row-space framework that represents a linear computing code by a global row space. This space must contain the target row space and be generated by rows satisfying the local support-constraints of the network. We prove that this representation is equivalent to the existence of a linear computing code. For any prescribed global row space, we give a necessary and sufficient condition for its realization and determine the minimum uniform communication load at the middle nodes. The condition is expressed in terms of the ranks of the local subspaces supported on the source-access sets. It separates the exact local realization problem from the outer problem of designing the global row space and yields a variational characterization of the linear computing capacity. We then apply the framework to MDS targets over cyclic networks. We identify when the target row space alone is sufficient and when auxiliary rows are required. We determine the capacity in the dense regime and in the sparse regime whenever the cut-set bound is integral. For the remaining sparse parameters, we give a general linear construction whose achievable rate equals the integer part of the cut-set bound.
\end{abstract}

\begin{IEEEkeywords}
Network function computation, vector-linear functions, three-layer network,  support-constraint, subspace covering, cyclic network.
\end{IEEEkeywords}

\section{Introduction}
%The rapid growth of large-scale distributed systems has made communication-efficient computation a central problem in information theory and coding theory. In many modern applications, including sensor networks, Internet-of-Things systems, federated learning, distributed analytics, and secure aggregation \cite{2004koetternetwork,2012rainetwork,2013ramamoorthycommunicating,2013AFKZ,2014Appuswamy,2018HTYG,2019GYYL}, a central terminal is often interested not in all raw data generated by distributed terminals, but only in a certain function of these data. Transmitting all source messages to the terminal and then computing the function is usually wasteful, especially when the network bandwidth is limited or when the number of sources is large. Network function computation provides a natural theoretical framework for understanding how much communication is fundamentally necessary when the objective is computation rather than message reconstruction. 

The rapid growth of large-scale distributed systems has made communication-efficient computation an important problem in information theory and coding theory. In many applications, including sensor networks, Internet-of-Things systems, federated learning, distributed analytics, and secure aggregation
\cite{2004koetternetwork,2012rainetwork,2013ramamoorthycommunicating,2013AFKZ,2014Appuswamy,2018HTYG,2019GYYL},
a central terminal needs only a function of the data generated at distributed terminals, rather than all raw data. Sending all source messages to the terminal before computing the function may therefore use much more communication than necessary \cite{2001Orlitsky}. Network function computation provides a general framework for studying the communication required for such tasks.

In a network function computation problem, a directed acyclic network $\cN$ contains multiple source nodes and one sink node. Each source node observes a message, and the sink wishes to compute a prescribed target function $f$ of all source messages. Intermediate nodes may encode their received symbols and forward them to downstream nodes.
We assume each link has unit capacity, which is a common and reasonable abstraction, especially when multiple parallel edges are allowed between nodes.
The main performance metric is the computing capacity, defined as the maximum average number of function instances that can be computed per use of the network. A fundamental question is how the computing capacity depends jointly on the network topology and the target function \cite{2011AFKZ,2018HTYG,2019GYYL,2014Appuswamy,2024GuangZhang}.

%This question has been studied extensively since the foundational work of Appuswamy \emph{et al.}~\cite{2011AFKZ,2013AFKZ}. General cut-set type upper bounds were developed for arbitrary networks and target functions, and these bounds were later refined and strengthened in~\cite{2018HTYG,2019GYYL,2024GuangZhang}. For several important cases, such as the identity function, the finite-field sum function, and some special network topologies, the known upper bounds are tight and can be achieved by linear network codes. However, in general, computing capacity remains difficult to characterize. Existing converse bounds are not always tight, and even when a meaningful upper bound is available, constructing a network code that achieves it can be highly nontrivial.

This question has been studied since the work of Appuswamy \emph{et al.}~\cite{2011AFKZ,2013AFKZ}. Cut-set upper bounds were developed for general networks and target functions and were later strengthened in \cite{2018HTYG,2019GYYL,2024GuangZhang}. These bounds are tight for several important cases, including the identity function, the finite-field sum function, and some special network topologies. In general, however, the computing capacity remains difficult to determine. Existing converse bounds are not always tight, and constructing a code that achieves optimal or nearly optimal computing rate can be nontrivial.

Among all target functions, vector-linear functions form one of the most important classes. A vector-linear target function asks the sink to compute several prescribed linear combinations of the source messages over a finite field. This setting includes both the identity function and the sum function as special cases, but it is substantially more general. It also naturally models a variety of distributed linear-processing tasks, such as linearly separable computation, coded distributed computation, and aggregation of multiple correlated statistics \cite{2021LSCWan,2022LSCITW,2022LSCWan,2026LSCcyclic,namboodiri2026fundamental,2025Tessellated,namboodiri2025fundamental,2018LiMa,2023Elia,2022WanSun,2024Elia,jafarpisheh2026secure,2025ChenCheng}. In these applications, the coefficient matrix of the desired linear map is not merely a technical object, its algebraic structure directly affects what can be computed locally and what must be communicated globally.

Despite their naturality, vector-linear target functions are much less understood than scalar sum or identity functions. Prior work has obtained general necessary conditions and capacity upper bounds for vector-linear function computation over arbitrary networks~\cite{2011AFKZ,2013AFKZ,2024ZhouFu}. Exact capacity results, however, are known only for rather restricted network classes, such as diamond networks~\cite{2022LiXudiamond} or small three-layer networks~\cite{guang2025distributed}. 
A main difficulty is that code feasibility depends not only on the rank of the target matrix but also on how its row space interacts with the source sets observed at the intermediate nodes. This interaction does not arise in the same form in routing or multicast, but it is central to vector-linear function computation.

This paper focuses on three-layer networks. A three-layer network consists of a source layer, a middle layer, and a single sink. Each middle node observes a prescribed subset of the sources and sends coded information to the sink. This architecture is not only a tractable special topology.  It is the basic communication interface in systems in which workers store different subsets of a dataset, relays have different sensing or access neighborhoods, or an aggregator collects locally formed linear statistics. Cyclic and heterogeneous data assignments in distributed linearly separable computation \cite{2021LSCWan,2022LSCITW,2022LSCWan,2026LSCcyclic,2026LSC_heterogeneous}, access networks in federated learning \cite{2024lindifferential,egger2023private}, and relay-based secure aggregation \cite{2025HSAcyclic,2024zhangwansunwangoptimal,2024ZhangWanSunWangITW,2024Luchengkangliucapacity,2025HSAxu} all lead naturally to three-layer source-access models. Thus, three-layer networks provide a basic but useful setting for isolating the algebraic constraints in vector-linear computation. 

Following the network function computation model, we treat both the network and the target function as fixed parts of the problem. Specifically, the source-access system $\Sigma=(\Sigma_1,\ldots,\Sigma_m)$ and the full-row-rank target matrix $T\in\F_q^{k\times s}$ are prescribed before the code is chosen. We study the zero-error linear computing capacity of the pair $(T,\Sigma)$. Neither the target coefficients nor the source assignment can be redesigned to improve the rate.
The fixed pair $(T,\Sigma)$ imposes two different requirements. Every row transmitted by a middle node must be supported on the source set observed by that node, whereas the information received at the sink must contain the target row space. Thus, code design must satisfy both the support-constraints imposed by the topology and the row-space requirement imposed by the target. Existing upper bounds identify communication bottlenecks, but they do not by themselves show how to choose local coding matrices that meet these two requirements.

To address this issue, we formulate linear computation as a support-constrained row-space realization problem. Rather than designing the local encoders separately, we study the global matrix received by the sink. Its row blocks must satisfy the support-constraints of the corresponding middle nodes, and its row space must contain the target row space. We first prove that this matrix formulation is equivalent to the original linear coding problem.
We then solve the realization problem for any prescribed global row space. Specifically, we give a necessary and sufficient condition for its realization and determine the minimum uniform communication load at the middle nodes. The condition depends only on the ranks of the locally supported parts of the prescribed space. Its proof uses a capacitated linear-space selection argument related to the linear Rado theorem \cite{rado1957note}. Optimizing the prescribed-space result over the block length and all admissible global row spaces gives a variational characterization of the linear computing capacity. This separates an exact local realization problem from an outer row-space design problem, which depends on the target and the topology.
This separation also explains the role of auxiliary rows. Restricting the global space to the target row space may provide too few locally supported directions. Enlarging the global space increases its dimension, but it may also create more locally supported directions. The achievable rate is determined by the balance between these two effects. The auxiliary directions are not part of the desired output and are removed by the sink decoder.

The closest application-level model is linearly separable computation (LSC) with one master and no stragglers. Datasets correspond to sources, workers to middle nodes, and the master to the sink. Message subdivision corresponds to the instance length, and worker transmissions correspond to the row-block loads. The two models therefore share the same encoder--decoder factorization and support-constraints. Their problem settings, however, are different.
Some LSC results consider a fixed cyclic or heterogeneous assignment and a generic random demand, often together with straggler constraints
\cite{2021LSCWan,2026LSC_heterogeneous}. Such results describe almost all demand matrices and do not distinguish targets with special algebraic structure. Other formulations fix the requested functions but allow the data assignment, encoding supports, or connectivity to be designed
\cite{2022LSCITW,2025Tessellated,namboodiri2025fundamental}. In the present work, both the finite-field target $T$ and the access system $\Sigma$ are prescribed, and zero-error recovery is required for this particular pair. The designing variables are only the linear code and its block length. Thus, our setting is the fixed-target, fixed-topology capacity problem from network function computation. Section~\ref{sec:prelim} gives a detailed comparison of the models and parameters.

The main contributions of this paper are summarized as follows.
\begin{enumerate}
\item We establish a support-constrained row-space representation for vector-linear computation over three-layer networks. 
The key observation is that every linear computing code is equivalent to constructing a received matrix whose row blocks satisfy the local support-constraints and whose row space contains the target row space. 
This reformulation provides a unified view of linear code design through global row-space realization.

\item We completely solve the realization problem for a prescribed global row space. 
We show that a row space can be realized if and only if its locally supported sections satisfy a collection of rank conditions, which also determines the minimum communication load. 
Beyond evaluating a given row space, this characterization reveals how enlarging the row space can improve local rank supply and provides guidance for the design of useful auxiliary directions.

\item We develop deterministic constructions for MDS targets over cyclic source-access systems by leveraging the dual viewpoint of linear codes. 
In particular, we design cyclic parity-check structures whose local null spaces provide the required support-constrained transmissions. 
This approach achieves the capacity upper bound in the dense regime and the sparse regimes with integral cut-set bounds, and achieves the integer part of the bound in the remaining sparse regimes.
\end{enumerate}

The rest of this paper is organized as follows. 
Section~\ref{sec:prelim} introduces the basic model of the vector-linear computation problem over three-layer networks, and reviews related results in network function computation and linearly separable computation.  
Section~\ref{sec:framework} develops the support-constrained row-space realization framework and shows how linear computation codes can be represented through a support-constrained matrix. 
Section~\ref{sec:linear} studies the maximum achievable rate for a fixed global row space and analyzes the gap between the upper bound and the achievable rate. 
Section~\ref{sec:mds_target} specializes the general framework to MDS target matrices and cyclic support systems, deriving deterministic constructions in both dense and below-threshold regimes. Finally, Section~\ref{sec:conclusion} concludes the paper and offers future directions.

\emph{Notations.} 
For a positive integer $a$, let $[a]\triangleq\{1,2,\ldots,a\}$. For integers $a\le b$, let $[a:b]\triangleq\{a,a+1,\ldots,b\}$. 
For a positive integer \(s\), indices in cyclic expressions are interpreted modulo \(s\). For integers \(a\) and \(r\geq 1\), define the cyclic interval of length \(r\) by $[a:a+r-1]_{\rm s} \triangleq \{((a-1+t)\bmod s)+1:t=0,1,\ldots,r-1\}$. When \(r\leq s\), this set contains \(r\) distinct elements. All vector spaces are over a finite field $\F_q$. For a matrix $A$, $\Rank(A)$ denotes its rank, and $\Row(A)$ denotes its row space. For a set $S$, $|S|$ denotes its cardinality. For a matrix $T\in\F_q^{k\times s}$ and a subset $A\subseteq[s]$, let $T_A$ denote the submatrix of $T$ formed by the columns indexed by $A$. 
For a vector $\bz\in\F_q^s$, its support is  $\supp(\bz)=\{j\in[s]:z_j\neq 0\}.$

\section{Model and Preliminaries}\label{sec:prelim}

In this section, we first recall the network function computation model and two known capacity upper bounds. We then introduce the three-layer source-access model and define linear computing codes. Finally, we review the most closely related results in network function computation and linearly separable computation.

\subsection{Network function computation and capacity upper bounds}
A general network function computation problem is defined over a directed acyclic graph $G=(\cV,\cE)$,  where multiple edges between two vertices are allowed. For an edge $e\in\cE$, let $\tail(e)$ and $\head(e)$ denote its \emph{tail} and \emph{head} nodes. For a vertex $v\in\cV$, define $\In(v)=\{e\in\cE:\head(e)=v\}, \Out(v)=\{e\in\cE:\tail(e)=v\}$.  A network is denoted by $\cN=(G,S,\gamma)$,  where $S=\{\sigma_1,\sigma_2,\ldots,\sigma_s\}$ is the set of source nodes and $\gamma$ is the sink node. The sink wishes to compute a \emph{target function} \[ f:\cA^s\to\cO \] of the source messages. \emph{An $(\ell,n)$ network code} computes \(\ell\) independent instances of \(f\) while transmitting at most \(n\) symbols on each edge. When the source and edge alphabets are the same, the computation rate is \(\ell/n\). The computing capacity \(C(\cN,f)\) is the supremum of all achievable rates. 

This paper focuses on vector-linear target functions over $\F_q$. Namely, the target function has the form  $f(x)=Tx$, where $T\in\F_q^{k\times s},\Rank(T)=k$. We recall two cut-set upper bounds for computing vector-linear target functions over a general network. These bounds will be used later as benchmarks for the three-layer constructions developed in this paper.
We first introduce the cut-set notation. For an edge subset $C\subseteq\cE$, let \[ I_C \triangleq \{i\in[s]: \text{ after deleting } C,\text{ there is no directed path from } \sigma_i \text{ to } \gamma\}. \] Thus, $I_C$ is the set of source indices separated from the sink by $C$. We say that $C$ is a \emph{cut set} if $I_C\neq\emptyset$, and denote by $\Lambda(\cN)$ the collection of all cut sets of $\cN$. 
The standard cut-set bound in \cite{2011AFKZ} specializes to the following rank bound.

\begin{thm}\label{thm:cut-rank-bound} For a network $\cN$ and a vector-linear target function $f(x)=Tx$, the computing capacity satisfies \[ C(\cN,T) \le \min_{C\in\Lambda(\cN)} \frac{|C|}{\Rank(T_{I_C})}. \] 
\end{thm}

A stronger upper bound is obtained by partitioning a cut set. For a cut set $C\in\Lambda(\cN)$, define \[ K_C \triangleq \{i\in[s]: \text{ there exists a directed path from }\sigma_i \text{ to } \tail(e)\text{ for some }e\in C\}. \] Clearly, $I_C\subseteq K_C$. 
We also write $ J_C\triangleq K_C\setminus I_C.$ 

\begin{dfn}[Definition~2 \cite{2019GYYL}] \label{def:strong-partition} Let $C\in\Lambda(\cN)$ be a cut set. A partition $\cP_C=\{C_1,C_2,\ldots,C_t\}$ of $C$ is called a strong partition of $C$ if the following two conditions hold: 
\begin{enumerate} 
\item $I_{C_a}\neq\emptyset$ for every $a\in[t]$; 
\item $I_{C_a}\cap K_{C_b}=\emptyset$ for every $a,b\in[t]$ with $a\neq b$. 
\end{enumerate} 
\end{dfn} 
The trivial partition $\{C\}$ is always a strong partition of $C$. For a strong partition $\cP_C=\{C_1,C_2,\ldots,C_t\}$,  define $I_a\triangleq I_{C_a}, a\in[t]$,  and $I_{\cP_C}\triangleq \bigcup_{a=1}^t I_a. $ The vector-linear rank associated with $\cP_C$ is defined as 
\begin{equation}\label{eq:RTPC}
    R_T(\cP_C) \triangleq \Rank(T_{I_C}) + \sum_{a=1}^t \Rank(T_{I_a}) - \Rank(T_{I_{\cP_C}}).
\end{equation}
Equivalently, the same quantity is sometimes denoted by $\Rank_{\cP_C}(T)$ in the literature. Using the definition of strong partition and the refined upper bound in \cite{2019GYYL}, the following upper bound is obtained.

\begin{thm}[Theorem~2 \cite{2024ZhouFu}]\label{thm:strong-partition-bound} For the network $\cN$ and the vector-linear target function $f(x)=Tx$, the computing capacity satisfies 
\[ C(\cN,T) \le \min_{C\in\Lambda(\cN)} \frac{|C|} {\max_{all\ \cP_C} R_T(\cP_C)}.\]
\end{thm}

These bounds will be used only through their specializations to the three-layer networks considered below.

\subsection{The three-layer source-access model and linear network codes}
\begin{figure}
    \centering
    \includegraphics[width=0.5\linewidth]{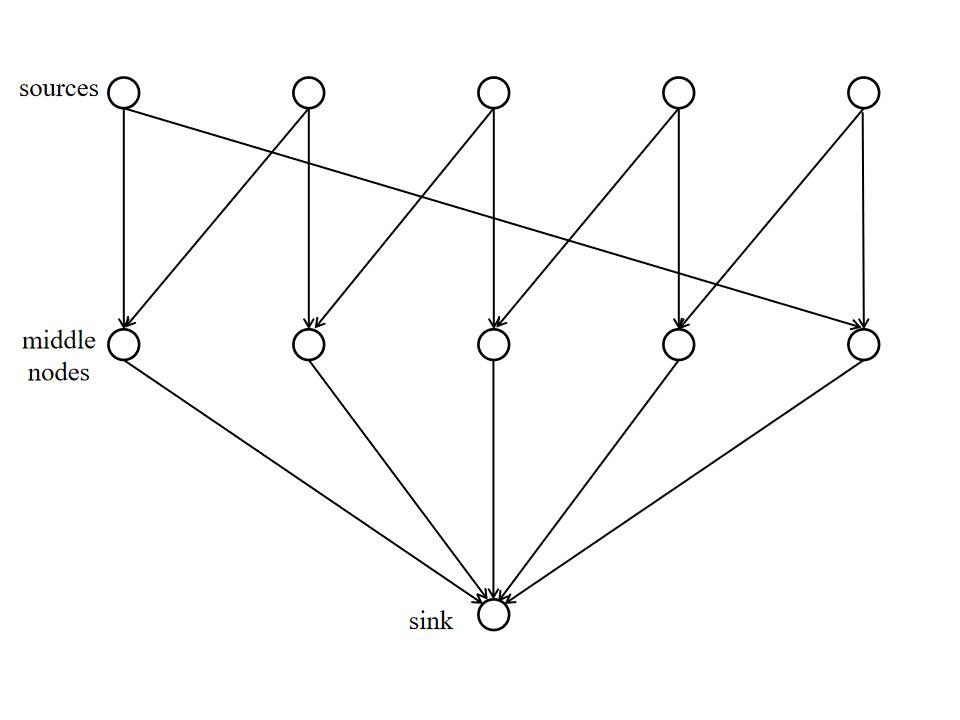}
    \caption{A three-layer network with $s=m=5$ and $\Sigma_i=[i:i+1]_{\rm 5}$.}
    \label{fig:3_layer_model}
\end{figure}
A three-layer network is denoted by  $\cN_{s,m,\Sigma}$,  where there are $s$ source nodes $\{\sigma_1,\sigma_2,\ldots,\sigma_s\},$ $m$ middle nodes $\{v_1,v_2,\ldots,v_m\}, $ and one sink node $\gamma$. The source-access topology is described by a set system 
\[ \Sigma=\{\Sigma_i:i\in[m]\},\qquad \Sigma_i\subseteq[s]. \] See Fig.~\ref{fig:3_layer_model} as an example.
There is an edge from \(\sigma_j\) to \(v_i\) if and only if \(j\in\Sigma_i\), and every middle node \(v_i\) has one outgoing edge to the sink. Thus, \(v_i\) can receive coded information only from the sources indexed by \(\Sigma_i\).

We now specialize the two cut-set bounds in
Theorems~\ref{thm:cut-rank-bound} and~\ref{thm:strong-partition-bound}
to the three-layer source-access network $\cN_{s,m,\Sigma}$.
For a source subset
$A\subseteq[s]$, define its middle-layer neighborhood by
\[
    \Gamma(A)\triangleq
    \{i\in[m]:\Sigma_i\cap A\neq\emptyset\}.
\]
For a single source index $j\in[s]$, we write $\Gamma(j)\triangleq\Gamma(\{j\})$.
We also define the missing source set at \(v_i\) as
$ \overline{\Sigma}_i\triangleq [s]\setminus \Sigma_i.$
For $B\subseteq[m]$, let $e_B\triangleq\{v_i\gamma:i\in B\}$ be the corresponding set of middle-to-sink edges.
For convenience, write $I_B=I_{e_B},K_B=K_{e_B}$. The three-layer topology gives
\[
    I_B=
    \{j\in[s]:\Gamma(j)\subseteq B\}\text{ and } K_B
    \triangleq
    \bigcup_{i\in B}\Sigma_i.
\]

For this network, the standard cut-set bound can be written as
\begin{equation}\label{eq:three-layer-cut-bound-B}
    C(\cN_{s,m,\Sigma},T)
    \le
    \min_{\substack{B\subseteq[m]\\ I_B\neq\emptyset}}
    \frac{|B|}{\Rank(T_{I_B})}.
\end{equation}
Equivalently, 
\[
    C(\cN_{s,m,\Sigma},T)
    \le
    \min_{\emptyset\neq A\subseteq[s]}
    \frac{|\Gamma(A)|}{\Rank(T_A)}.
\]

We next record the middle-layer specialization of the strong-partition bound. This specialization considers only cuts formed by middle-to-sink edges.
A partition $\cP_B=\{B_1,B_2,\ldots,B_t\}$ of $B$ is called a strong partition of $B$ if 
\begin{itemize}
    \item $I_{B_a}\neq\emptyset,\forall a\in[t]$;
    \item $I_{B_a}\cap K_{B_b}=\emptyset,\forall a\neq b.$
\end{itemize}
Let $\mathcal{P}(B)$ denote the collection of all such partitions.
For \(\cP_B\in\mathfrak P(B)\),  define $I_{\cP_B}\triangleq \bigcup_{a=1}^t I_{B_a}$ and
\[
    R_T(\cP_B)
    \triangleq
    \Rank(T_{I_B})
    +
    \sum_{a=1}^t \Rank(T_{I_{B_a}})
    -
    \Rank(T_{I_{\cP_B}}).
\]
Then, the strong partition cut-set bound gives
\begin{equation}\label{eq:three-layer-strong-bound}
    C(\cN_{s,m,\Sigma},T)
    \le
    \min_{\substack{B\subseteq[m]\\ I_B\neq\emptyset}}
    \frac{|B|}
    {\max_{\cP_B\in\mathcal{P}(B)}R_T(\cP_B)}.
\end{equation}
Since \(\{B\}\) is a strong partition and $R_T(\{B\})=\Rank(T_{I_B})$, the bound in \eqref{eq:three-layer-strong-bound} is at least as strong as \eqref{eq:three-layer-cut-bound-B}.

\subsection{Linear network computing codes}
For a positive integer $\ell$, let the $j$-th source generate $\bx_j=(x_{j,1},x_{j,2},\ldots,x_{j,\ell})^T\in\F_q^\ell.$  
For each instance $a\in[\ell]$, define \[ \bx^{(a)}=(x_{1,a},x_{2,a},\ldots,x_{s,a})^T\in\F_q^s. \] 
We order the complete source vector as \[ \bx_S=\big((\bx^{(1)})^T,(\bx^{(2)})^T,\ldots,(\bx^{(\ell)})^T\big)^T \in\F_q^{s\ell}. \] 
The desired output for the \(\ell\) instances is therefore  $(I_\ell\otimes T)\bx_S\in \F_q^{k\ell}$,  where $\otimes$ denotes the Kronecker product.

A linear $(\ell,n)$ computation code for $(\cN_{s,m,\Sigma},T)$ consists of the following linear maps. 
\begin{itemize} 
\item For every \(j\in\Sigma_i\), source \(\sigma_j\) sends $\bz_{j,i}=A_{j,i}\bx_j$ to middle node \(v_i\), where $A_{j,i}\in\F_q^{n\times\ell}$.
\item Middle node \(v_i\) sends $\by_i=\sum_{j\in\Sigma_i}B_{j,i}\bz_{j,i}$ to the sink, where $B_{j,i}\in\F_q^{n\times n}$.
\item The sink applies a decoding matrix $D\in\F_q^{k\ell\times mn}$,  which maps the received vector  $\by=(\by_1^T,\by_2^T,\ldots,\by_m^T)^T\in\F_q^{mn}$  to the desired target vector. 
\end{itemize} 
The code is valid if, for every source message vector $\bx_S\in\F_q^{s\ell}$, \[ D\by=(I_\ell\otimes T)\bx_S. \] Its computation rate is $R=\frac{\ell}{n}. $
The linear computing capacity of $(\cN_{s,m,\Sigma},T)$ is defined as \[ C_{\rm lin}(\cN_{s,m,\Sigma},T) \triangleq \sup\left\{\frac{\ell}{n}: \text{there exists a valid linear }(\ell,n)\text{ computation code}\right\}. \]
Clearly,
\[
    C_{\mathrm{lin}}(\cN_{s,m,\Sigma},T)
    \leq
    C(\cN_{s,m,\Sigma},T).
\]

\begin{rmk}[On source-to-middle transmissions] \label{rem:source-to-middle} 
In the definition of $(\ell,n)$ network codes, a middle node is \emph{not} assumed to know the entire $\ell$-symbol message of each adjacent source. Each edge $\sigma_jv_i$ carries only $n$ coded symbols. However, this is sufficient for our matrix formulation. Indeed, if middle node $v_i$ is supposed to transmit an $n$-dimensional linear combination involving source $x_j$, then source $\sigma_j$ can send exactly the corresponding $n$-dimensional contribution to $v_i$. Thus the source-to-middle layer does not impose an additional restriction beyond the support-constraint determined by $\Sigma_i$. 
\end{rmk}

\subsection{Related models and prior results}

\subsubsection{Vector-linear network function computation}

For vector-linear target functions, Appuswamy and Franceschetti \cite{2014Appuswamy} studied the existence of rate-one linear network codes. 
One important case considered there is when the target matrix has rank $s-1$ and is equivalent to a matrix of the form
$\begin{bmatrix}
    I_{s-1}\ \bu
\end{bmatrix}$,
where all entries of $\bu$ are nonzero. 
Their result shows that, under this nondegeneracy condition, the cut set bound can guarantee the existence of a rate-one linear code. 
This result already reveals that the algebraic structure of the target matrix is crucial. Specifically, two target matrices with the same rank may behave differently if their row or column dependencies interact differently with the network topology.

Exact capacity results are known for several restricted topologies.
Li and Xu \cite{2022LiXudiamond} characterized the capacity of vector-linear function computation over diamond networks. The three-source model studied in \cite{guang2025distributed} is closer to the topology considered here. In \cite{guang2025distributed}, several encoders observe prescribed source subsets and communicate with one decoder. That work gives detailed results for three sources and shows that the strong-partition upper bound need not be tight even for small vector-linear instances.

The present paper considers arbitrary numbers of sources and middle nodes, while focusing on linear codes and structured target matrices. Its purpose is not to classify all three-layer instances, but to develop a general support-constrained realization method and to obtain deterministic results for important structured families.

\subsubsection{Linearly separable computation}

Linearly separable computation (LSC) is the closest distributed computing model to the three-layer problem studied here. To avoid introducing a second notation system, we use the notation of this paper throughout the comparison. In the common LSC notation, the correspondence is
\[K=s,\qquad N=m,\qquad K_c=k,\]
and the LSC demand matrix \(F\) corresponds to our target matrix \(T\).

In an LSC problem, dataset \(j\in[s]\) produces a message \(\bx_j\), and worker \(i\in[m]\) is assigned the datasets indexed by \(\Sigma_i\). The master wishes to recover the \(k\) linear combinations specified by \(T\). If each message contains \(\ell\) symbols and worker \(i\) transmits \(p_i\) symbols, its normalized communication load is
\[\lambda_i\triangleq\frac{p_i}{\ell}.\]
The maximum worker load and total download are, respectively,
\[\lambda_{\max}\triangleq\max_{i\in[m]}\lambda_i,\qquad\lambda_{\mathrm{sum}}\triangleq\sum_{i=1}^m\lambda_i.\]
Under the uniform-load restriction \(p_i=n\), we have
\[
    \lambda_{\max}=\frac{n}{\ell}=\frac{1}{R},
    \qquad
    \lambda_{\mathrm{sum}}=\frac{mn}{\ell}=\frac{m}{R}.
\]
Thus, in the no-straggler and uniform-load setting, LSC and the present model have the same linear encoder--decoder feasibility problem. Their theorems may nevertheless address different demand models, assignment variables, block models, and performance measures.

A straggler-robust LSC scheme requires the master to decode from every responding set of a prescribed size. The present network model has no stragglers, that is, the sink receives the transmissions of all \(m\) middle nodes. Moreover, the unit-capacity middle-to-sink edges lead naturally to a maximum per-node load, rather than only a total download constraint.
Existing LSC work most closely related to our setting can be organized into two complementary directions, according to whether the data assignment or the demand is treated as the principal designing variable.

\textbf{Fixed assignment and generic demand}:
A common LSC assumption is that the entries of \(T\) are chosen independently and uniformly from a sufficiently large finite field. We refer to this as a \emph{generic demand}. For any fixed \(A\subseteq[s]\), such a matrix satisfies \[\Rank(T_A)=\min\{k,|A|\}\] with high probability as the field size grows. Generic demands therefore behave like MDS targets with high probability. A generic-demand theorem guarantees performance for almost all target matrices, but it need not apply to a particular structured target.

The information-theoretic LSC formulation in \cite{2021LSCWan,2022LSCWan} studies communication--computation tradeoffs under this demand model. 
For cyclic assignments, Huang et al.~\cite{2026LSCcyclic} introduce virtual demands and use interference alignment to eliminate the messages unavailable at each worker. Their demand matrix, virtual demands, and several encoding coefficients are selected randomly over a sufficiently large field, and decodability is formulated as a high-probability guarantee as the field size tends to infinity. Under the no-straggler specialization with equal numbers of datasets and workers, their low-demand-dimension regime corresponds to the lower-support regime \(r<s-k\) considered in this paper. For this regime, analytic decodability is established for specified parameter cases, while additional finite lengths are verified numerically. In contrast, our parity-check construction is deterministic and gives a complete algebraic guarantee for every prescribed MDS target. Moreover, the construction works over any field satisfying \(q>r+k-1\). Thus, it provides a deterministic finite-field guarantee depending only on the sliding-window length, without requiring asymptotically large fields or random full-rank events.

The heterogeneous model in \cite{2026LSC_heterogeneous} allows the assignment \(\Sigma\) to be arbitrary and given in advance. It studies the largest generic-demand dimension that can be supported under a communication-load constraint. Its converse and construction are described through zero submatrices of the assignment pattern. To make this concept precise, let $P_\Sigma\in\{0,1\}^{m\times s}$ with
\[
    (P_\Sigma)_{i,j}
    =
    \begin{cases}
    1, & j\in\Sigma_i,\\
    0, & j\notin\Sigma_i.
    \end{cases}
\]
For \(G\subseteq[m]\) and \(Q\subseteq[s]\), the submatrix \(P_\Sigma(G,Q)\) is a \emph{zero submatrix} if
\[P_\Sigma(G,Q)=\Zero.\]
Equivalently, none of the workers in \(G\) has access to any dataset in \(Q\). Such a zero submatrix limits the number of independent demand dimensions involving the messages indexed by \(Q\) that can be supplied by the workers outside \(G\). The results in \cite{2026LSC_heterogeneous} give universal converse and achievability bounds based on this structure, with equality in specified regimes, and extend the bounds from integer to fractional worker loads.

For a generic demand, the restriction to \(Q\) has rank \(\min\{k,|Q|\}\) with high probability. For a prescribed target, however, the relevant quantity is its actual rank \(\Rank(T_Q)\). The following example shows that even the converse of \cite{2026LSC_heterogeneous} cannot be transferred unchanged from a generic demand to a prescribed structured demand.

\begin{figure}
    \centering
    \includegraphics[width=0.5\linewidth]{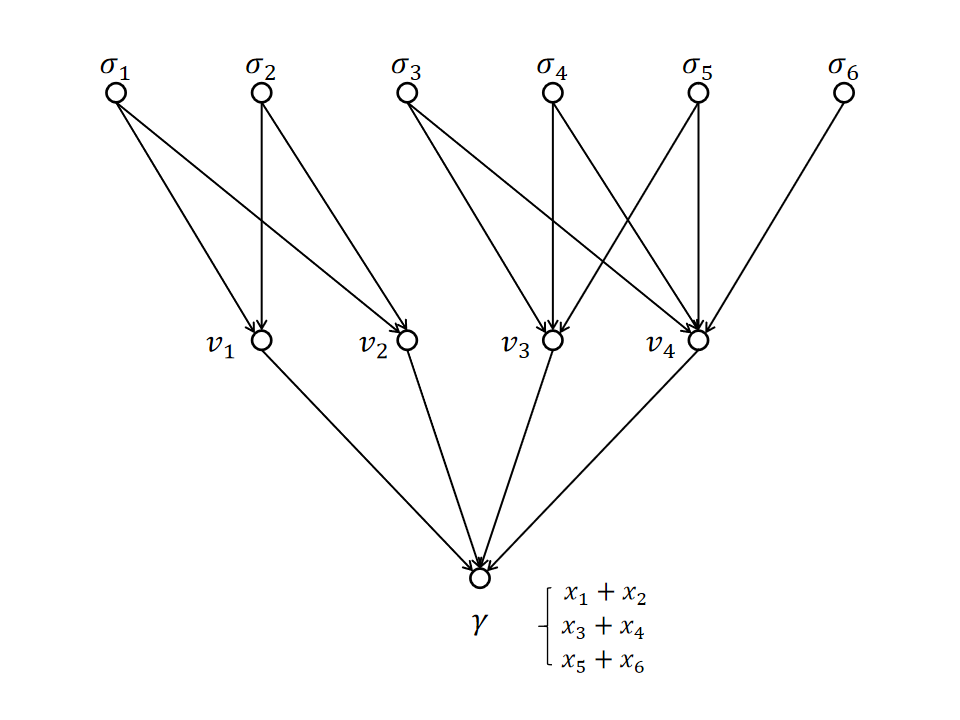}
    \caption{A three-layer network with $\Sigma_1=\Sigma_2=\{1,2\},\
    \Sigma_3=\{3,4,5\},$ and $
    \Sigma_4=\{3,4,5,6\}$.}
    \label{fig:LSC_example}
\end{figure}

\begin{example}\label{ex:generic-structured-lsc}
Let \(s=6\), \(m=4\), and
\[
    \Sigma_1=\Sigma_2=\{1,2\},\qquad
    \Sigma_3=\{3,4,5\},\qquad
    \Sigma_4=\{3,4,5,6\}.\]
The network is depicted in Fig.~\ref{fig:LSC_example}. Suppose that each worker sends at most one symbol for one function instance. Consider $ G=\{1,2\}$ and $Q=\{3,4,5,6\}.$ Since \(P_\Sigma(G,Q)=\Zero\), all information involving the messages in \(Q\) must be supplied by workers \(3\) and \(4\). These two workers send only two symbols in total. A generic rank-three target satisfies
\[
    \Rank(T_Q)=3
\]
with high probability and therefore cannot be computed under this load.

Now consider the prescribed groupwise-sum target
\[
    T
    =
    \begin{bmatrix}
    1&1&0&0&0&0\\
    0&0&1&1&0&0\\
    0&0&0&0&1&1
    \end{bmatrix}.
\]
Although \(\Rank(T)=3\), its restriction to \(Q\)
has rank
\[
    \Rank(T_Q)=2.
\]
Worker \(1\) can send \(x_1+x_2\), worker \(3\) can send \(x_3+x_4\), and worker \(4\) can send \(x_5+x_6\). The sink then recovers all three requested functions.

Thus, a generic rank-three demand is impossible under this assignment and load, whereas a particular rank-three target is feasible. The fixed-target problem must therefore retain the restricted ranks and row-space structure of \(T\).
\end{example}

\textbf{Demand-aware and universal assignment design:}
A second line of LSC research starts from a prescribed demand matrix and designs the task assignment or communication pattern. A basic algebraic representation is the sparse factorization $T=DE$ or, with message subdivision,
\[
    I_\ell\otimes T=DE.
\]
The rows of \(E\) describe the linear combinations computed by the workers, while \(D\) describes how these combinations are decoded. Row sparsity in \(E\) controls which datasets must be assigned to each worker. In multi-user models, sparsity in \(D\) also controls which worker--user communication links are used.

Khalesi and Elia \cite{2022LSCITW} relate this factorization problem to syndrome decoding, covering codes, and partial covering codes. The task placement and the supports of the factors are designing variables for the requested functions. Tessellated distributed computing \cite{2025Tessellated} develops a related sparse factorization model over the real field and also permits approximate recovery.

A \emph{universal demand} requirement asks one assignment and delivery design to support every demand matrix in a prescribed class, rather than one fixed target or a randomly drawn target. The nullspace-based construction in \cite{namboodiri2025fundamental} jointly designs task assignment and server transmissions for such a requirement. Under linear encoding and no subpacketization, it obtains exact tradeoffs in some regimes and constant-factor results in general. Here, no subpacketization means that each component message is treated as one coding block; allowing message subdivision corresponds to using a nontrivial block length such as \(\ell>1\) in our model. The multi-user extension \cite{namboodiri2026fundamental} additionally constrains the worker-user communication pattern.

These factorization and nullspace viewpoints are closely related to the constructions in this paper. In particular, the virtual demands used in cyclic LSC play a role similar to auxiliary rows: they enlarge an effective demand space so that locally supported transmissions become available. Therefore, neither the factorization \(T=DE\) nor the introduction of auxiliary demand rows is, by itself, the distinction of the present work.

The difference lies in which objects are fixed and which problem is solved. In our model, both \(T\) and \(\Sigma\) are prescribed. Hence the support pattern of the encoding matrix cannot be redesigned after the target is given. We allow arbitrary block length and require zero-error recovery for this particular finite-field pair \((T,\Sigma)\). Our first result proves that every linear code for this pair is equivalent to a support-constrained matrix realization. Our second result then gives a necessary and sufficient condition for realizing any prescribed global row space under the fixed support system \(\Sigma\).

\section{Support-Constrained Matrix Representation of Linear Codes}\label{sec:framework}
This section gives a matrix representation of linear computation codes over the three-layer network. The main idea is to collect all symbols received by the sink into one matrix. The access sets constrain the support of its row blocks, while decodability imposes a condition on its row space. This algebraic viewpoint is related to the matrix-based formulation
of network coding introduced in~\cite{2003Koetter}. This separates the local constraints of the network from the global requirement of the target function.

For \(j\in[s]\), define
\[ J_j^{(\ell)} \triangleq \{(a-1)s+j:a\in[\ell]\}\subseteq[s\ell].\]
These are the coordinates of \(\bx_S\) that belong to source \(\sigma_j\).
For \(A\subseteq[s]\), let
\[
    J_A^{(\ell)}
    \triangleq
    \bigcup_{j\in A}J_j^{(\ell)}
\]
and define the corresponding coordinate subspace by
\[
    \mathscr E_A^{(\ell)}
    \triangleq
    \left\{
        \bu\in\F_q^{s\ell}:
        u_t=0\text{ for every }t\notin J_A^{(\ell)}
    \right\}.
\]
We also write $\mathscr C_T^{(\ell)}\triangleq\Row(I_\ell\otimes T)$ for the lifted target row space. Since \(T\) has full row rank, \(\dim\mathscr C_T^{(\ell)}=k\ell\).

For a linear \((\ell,n)\) code, the vector sent by middle node \(v_i\) can be written as
\[
    \by_i
    =
    \sum_{j\in\Sigma_i}B_{j,i}A_{j,i}\bx_j
    \triangleq
    E_i\bx_S,
\]
where $E_i\in\F_q^{n\times s\ell}.$ Because \(v_i\) has access only to the sources in \(\Sigma_i\), every row of \(E_i\) belongs to \(\cE_{\Sigma_i}^{(\ell)}\).

\begin{dfn}[Support-constrained matrix]\label{def:support-constrained-matrix}
A block matrix
\[
    E
    =
    \begin{bmatrix}
        E_1\\
        E_2\\
        \vdots\\
        E_m
    \end{bmatrix}
    \in\F_q^{mn\times s\ell},
\]
where $E_i\in\F_q^{n\times s\ell},$ is called \emph{\(\Sigma\)-support-constrained} if for every $i\in[m]$,
\[\Row(E_i)\subseteq\cE_{\Sigma_i}^{(\ell)}.\]
Equivalently, for every $i\in[m]$, $E_i\bigl(:,J_{[s]\setminus\Sigma_i}^{(\ell)}\bigr)=\Zero.$
\end{dfn}

The following theorem is the basic representation result.

\begin{thm}\label{thm:realization-implies-code}
There exists a linear $(\ell,n)$ computation code for $(\cN_{s,m,\Sigma},T)$ if and only if there exists a $\Sigma$-support-constrained matrix $E\in\F_q^{mn\times s\ell}$  such that
\begin{equation}
\label{eq:received-row-space-condition}
    \mathscr C_T^{(\ell)}\subseteq\Row(E).
\end{equation}
\end{thm}
\begin{pf} 
We first prove necessity. Suppose a valid linear $(\ell,n)$ computation code exists. As discussed above, the vector transmitted by middle node $v_i$ can be written as  $\by_i=E_i \bx_S$  for some matrix $E_i\in\F_q^{n\times s\ell}$ supported only on $J_{\Sigma_i}^{(\ell)}$. 
Stacking all middle-node transmissions gives \[ \by= \begin{bmatrix} \by_1\\ \vdots\\ \by_m \end{bmatrix} = E \bx_S, \text{ with } E= \begin{bmatrix} E_1\\ \vdots\\ E_m \end{bmatrix} \] being $\Sigma$-support-constrained. Since the sink can linearly decode the target function, there exists $D_{\rm dec}$ such that \[ DE=I_\ell\otimes T. \] 
Therefore, $\Row(I_\ell\otimes T)\subseteq \Row(E).$  

We now prove sufficiency. Suppose that a $\Sigma$-support-constrained matrix $E$ satisfies \[ \Row(I_\ell\otimes T)\subseteq \Row(E). \] Then there exists a decoding matrix $D$ such that $DE=I_\ell\otimes T.$  
It remains to show that $E$ can be implemented by a linear code over the three-layer network. 
For each middle node $v_i$, write the block $E_i$ as \[ E_i \bx_S = \sum_{j\in\Sigma_i} E_{i,j}\bx_j, \] where $E_{i,j}\in\F_q^{n\times\ell}$ is the submatrix of $E_i$ corresponding to the coordinates of source $\sigma_j$. 
This expression contains only $j\in\Sigma_i$ because $E_i$ is $\Sigma$-support-constrained. 
For each edge $\sigma_j\to v_i$ with $j\in\Sigma_i$, let source $\sigma_j$ send $\bz_{j,i}=E_{i,j}\bx_j\in\F_q^n.$  
This uses exactly $n$ symbols on the edge $\sigma_j\to v_i$. 
Middle node $v_i$ then computes \[ \by_i=\sum_{j\in\Sigma_i} \bz_{j,i} = E_i \bx_S \] and sends $\by_i$ to the sink. The sink applies $D$ and obtains \[ D\by = DE \bx_S = (I_\ell\otimes T)\bx_S. \] Hence a valid linear $(\ell,n)$ computation code exists. 
\end{pf}

Theorem~\ref{thm:realization-implies-code} shows that the source-to-middle layer creates no further algebraic constraint. Once a supported row block \(E_i\) is chosen, its contribution from each adjacent source can be sent on the corresponding source-to-middle edge. Hence the code design problem is fully described by two conditions on \(E\):
\begin{itemize}
    \item for every $i\in[m]$, $\Row(E_i)\subseteq\mathscr E_{\Sigma_i}^{(\ell)}$;
    \item $\mathscr C_T^{(\ell)}\subseteq\Row(E).$
\end{itemize}

This representation also identifies the global row space as the main design object. Let $\mathscr{W}\triangleq\Row(E)$ and $\mathscr U_i\triangleq\Row(E_i)$ for $i\in[m]$.
Then, 
\begin{equation}
\label{eq:global-local-space-conditions}
    \mathscr C_T^{(\ell)}\subseteq\mathscr W,
    \qquad\mathscr
    U_i\subseteq \mathscr W\cap\mathscr E_{\Sigma_i}^{(\ell)},
    \qquad
    \dim\mathscr U_i\le n,
    \qquad\mathscr
    W=\sum_{i=1}^m\mathscr U_i.
\end{equation}
Conversely, if a space \(\mathscr W\) and subspaces \(\mathscr U_1,\ldots,\mathscr U_m\) satisfy \eqref{eq:global-local-space-conditions}, choose at most \(n\) rows that span each \(\mathscr U_i\), and pad the corresponding block with zero rows when needed. The resulting matrix \(\mathscr E\) satisfies the conditions of Theorem~\ref{thm:realization-implies-code}. Thus, choosing a code is equivalent to choosing a global space \(W\) that contains the target space and can be generated from its locally supported parts.

We next give a normalized form that will be useful for deterministic constructions. If a code exists, then \(k\ell\le mn\), because the target space has dimension \(k\ell\) and \(E\) has only \(mn\) rows. For such \(\ell,n\), let
\[
    Q(V)
    \triangleq
    \begin{bmatrix}
        I_\ell\otimes T\\
        V
    \end{bmatrix}, 
\]
where $ V\in\F_q^{(mn-k\ell)\times s\ell}.$ The first \(k\ell\) rows are the target rows. The rows of \(V\) are auxiliary rows that may enlarge the global row space.

\begin{dfn}
\label{def:row-space-realization}
A pair \((V,P)\) with $V\in\F_q^{(mn-k\ell)\times s\ell}$ and $P\in\GL_{mn}(\F_q)$\footnote{$\GL_{mn}(\F_q)$ is the set of invertible matrices in $\F_q^{mn\times mn}.$} is called an \((\ell,n)\) \emph{support-constrained row-space realization} for \((T,\Sigma)\) if
\begin{equation}
\label{eq:auxiliary-row_form}
    E=PQ(V)
\end{equation}
is \(\Sigma\)-support-constrained.
\end{dfn}

\begin{cor}
\label{cor:normalized-row-space-realization}
For fixed \(\ell,n\) with \(k\ell\le mn\), a linear \((\ell,n)\) computation code exists if and only if there exists an \((\ell,n)\) support-constrained row-space realization \((V,P)\).
\end{cor}
\begin{pf}
If \((V,P)\) is a support-constrained row-space realization, then \(E=PQ(V)\) is \(\Sigma\)-support-constrained. Moreover,
\[
    \mathscr C_T^{(\ell)}
    \subseteq\Row(Q(V))
    =\Row(E),
\]
where the equality follows from the invertibility of \(P\). Theorem \ref{thm:realization-implies-code} therefore gives a linear code.

Conversely, let \(E\) be a matrix given by Theorem~\ref{thm:realization-implies-code}, and let \(\mathscr W=\Row(E)\). Since \(\mathscr C_T^{(\ell)}\subseteq\mathscr W\), extend the \(k\ell\) rows of \(I_\ell\otimes T\) to a basis of \(\mathscr W\). Place the added basis vectors in \(V\), and fill any remaining rows of \(V\) with zeros. Then
\[
    \Row(Q(V))=\mathscr W=\Row(E).
\]
The matrices \(Q(V)\) and \(E\) have the same number of rows and the same row space, so they are row equivalent. Hence there is an invertible matrix \(P\in\GL_{mn}(\F_q)\) such that \(E=PQ(V)\).
\end{pf}

The two matrices in \eqref{eq:auxiliary-row_form} have different roles. The matrix \(V\) selects the global space
\[
    \mathscr W=\Row(Q(V))
    =\mathscr C_T^{(\ell)}+\Row(V).
\]
The matrix \(P\) then selects the rows sent by the middle nodes. More precisely, let \(P_i\) be the \(n\)-row block of \(P\) associated with \(v_i\). The support requirement is
\begin{equation}
\label{eq:block-support-realization}
    P_iQ(V)
    \bigl(:,J_{[s]\setminus\Sigma_i}^{(\ell)}\bigr)
    =0,
    \qquad i\in[m].
\end{equation}
Thus, a deterministic construction may be described by first choosing auxiliary rows \(V\), and then finding an invertible \(P\) whose row blocks satisfy \eqref{eq:block-support-realization}.

The same form gives the decoding matrix directly. If the sink receives \(E\bx_S=PQ(V)\bx_S\), it first applies \(P^{-1}\) and then keeps the first \(k\ell\) coordinates. Thus, we may choose
\begin{equation}
\label{eq:normalized-realization-decoder}
    D
    =
    \begin{bmatrix}
        I_{k\ell}&\Zero_{k\ell\times(mn-k\ell)}
    \end{bmatrix}
    P^{-1}.
\end{equation}
Indeed,
\[
    DE
    =
    \begin{bmatrix}
        I_{k\ell}&\Zero_{k\ell\times(mn-k\ell)}
    \end{bmatrix}
    P^{-1}PQ(V)
    =I_\ell\otimes T.
\]
The auxiliary coordinates are discarded at this final step. In particular, taking all rows of \(V\) to be zero gives a direct realization using only the target row space. Nonzero auxiliary rows are useful when the target space alone does not contain enough locally supported directions.

The framework now separates the capacity problem into two parts. The outer
problem is to choose a global space \(\mathscr W\) containing
\(\mathscr C_T^{(\ell)}\). The inner problem is to decide whether \(\mathscr W\) can be
generated by at most \(n\) vectors from each local section $\mathscr W\cap \mathscr E_{\Sigma_i}^{(\ell)},$ for $i\in[m].$
The next section solves this inner problem for every fixed \(\mathscr W\) and then optimizes over the admissible global row spaces.

\section{Linear computing capacity through local row-space sections}\label{sec:linear}

Section~\ref{sec:framework} represents every linear computation code by a support-constrained received matrix \(E\). Its row space $\mathscr W=\Row(E)$ must contain the lifted target space \(\mathscr C_T^{(\ell)}\), while the rows sent by middle node \(v_i\) must be supported on \(\Sigma_i\). This section fixes the global space \(\mathscr W\) and asks a basic question: how many rows must each middle node send in order to generate \(\mathscr W\) from its locally supported parts?

We first solve this fixed-space problem exactly. We then apply the result to the smallest possible choice \(\mathscr W=\mathscr C_T^{(\ell)}\). This gives the best rate obtainable without auxiliary rows and shows when a larger global space is needed. The last subsection gives a rank-based rule for choosing such auxiliary rows.

\subsection{Exact realization of a fixed global row space}\label{subsec:fixed_auxiliary_rows}

Fix an instance length \(\ell\) and a subspace $\mathscr{W}$ with $\mathscr C_T^{(\ell)}\subseteq\mathscr  W\subseteq\F_q^{s\ell}.$ The rows available to middle node \(v_i\) are precisely the vectors in \(\mathscr W\) supported on the sources in \(\Sigma_i\). For every $v_i$, $i\in[m]$, define its \emph{local section} by
\begin{equation}
    \mathscr L_i(W)\triangleq\mathscr W\cap\mathscr E_{\Sigma_i}^{(\ell)}.
    \label{eq:local-row-space-section}
\end{equation}
For \(B\subseteq[m]\), define
\begin{equation}
    \rho_{\mathscr W}(B)\triangleq\dim\left(\sum_{i\in B}\mathscr L_i(\mathscr W)\right).\label{eq:local-row-space-rank}
\end{equation}
Thus, \(\rho_{\mathscr W}(B)\) is the number of independent directions in \(\mathscr W\) that the nodes in \(B\) can supply. We call it the \emph{local rank supply} of \(B\). Denote
\[
    d_{\mathscr W}\triangleq\dim\mathscr W.
\]

Suppose that every middle node sends at most \(n\) rows. The nodes in \(B\) can supply at most \(\rho_{\mathscr W}(B)\) dimensions, and the remaining \(m-|B|\) nodes can supply at most \((m-|B|)n\) dimensions. Hence every realization of \(\mathscr W\) must satisfy that for every $B\subseteq[m]$,
\begin{equation}
    \rho_{\mathscr W}(B)+(m-|B|)n\ge d_{\mathscr W}.
    \label{eq:fixed-space-necessary}
\end{equation}
The main point of this section is that these natural inequalities are also sufficient. Although local sections may overlap, they create no additional obstruction.

\begin{thm}[Fixed-space realization]\label{thm:fixed_auxiliary_space}
Let \(\mathscr W\subseteq\F_q^{s\ell}\) satisfy \(\mathscr C_T^{(\ell)}\subseteq \mathscr W\), and let \(d_{\mathscr W}=\dim\mathscr W\). For a positive integer \(n\), there exists a \(\Sigma\)-support-constrained matrix $ E\in\F_q^{mn\times s\ell}$ with \(\Row(E)=\mathscr W\) if and only if for every $B\subseteq[m]$,
\begin{equation}
    \rho_{\mathscr W}(B)+(m-|B|)n\ge d_{\mathscr W}.
    \label{eq:fixed-auxiliary-row}
\end{equation}
Whenever such a matrix exists, it defines a linear \((\ell,n)\) computation code for \((\cN_{s,m,\Sigma},T)\).
\end{thm}

Necessity follows from the dimension count above. To prove sufficiency, we use two linear-algebraic selection lemmas. The first is the linear form of Hall's condition.

\begin{lem}\label{lem:linear-rado}
Let $\mathscr{W}_1,\mathscr{W}_2,\ldots,\mathscr{W}_t$ be subspaces of a finite-dimensional vector space $\mathscr{W}$. There exist vectors $\bw_i\in \mathscr{W}_i,\forall i\in[t],$ such that $\bw_1,\bw_2,\ldots,\bw_t$ are linearly independent if and only if for all $J\subseteq[t]$,
\begin{equation}\label{eq:linear_rado}
    \dim\left(\sum_{i\in J}\mathscr{W}_i\right)\ge |J|.
\end{equation}
\end{lem}
\begin{pf}
The necessity is immediate. If $\bw_i\in \mathscr{W}_i$, $i\in[t]$, are linearly independent, then for every $J\subseteq[t]$ the space $\sum_{i\in J}\mathscr{W}_i$ contains the $|J|$ linearly independent vectors $\{\bw_i:i\in J\}$. Hence, \eqref{eq:linear_rado} holds.

For sufficiency, we use induction on \(t\). The claim is clear for \(t=1\). Assume that it holds for every family with fewer than \(t\) subspaces.

First suppose that there exists a nonempty proper subset $J\subsetneq[t]$ such that \[\dim\left(\sum_{i\in J}\mathscr{W}_i\right)=|J|.\]
Let $\mathscr{W}_J\triangleq \sum_{i\in J}\mathscr{W}_i.$ By induction, we can select linearly independent vectors $\bw_i\in \mathscr{W}_i$, $i\in J$. Since $\dim \mathscr{W}_J=|J|$, these vectors form a basis of $\mathscr{W}_J$.
Now consider the quotient space $\mathscr{W}/\mathscr{W}_J$. For each $i\in[t]\setminus J$, let $\overline{\mathscr{W}}_i\triangleq (\mathscr{W}_i+\mathscr{W}_J)/\mathscr{W}_J.$ For every $H\subseteq[t]\setminus J$, we have
\[
    \dim\left(\sum_{i\in H}\overline{\mathscr{W}}_i\right)
    =
    \dim\left(\mathscr{W}_J+\sum_{i\in H}\mathscr{W}_i\right)-\dim \mathscr{W}_J.
\]
By the assumed dimension condition for the original subspaces,
\[
    \dim\left(\mathscr{W}_J+\sum_{i\in H}\mathscr{W}_i\right)
    =
    \dim\left(\sum_{i\in J\cup H}\mathscr{W}_i\right)
    \ge |J|+|H|.
\]
Since $\dim \mathscr{W}_J=|J|$, it follows that
\[
    \dim\left(\sum_{i\in H}\overline{\mathscr{W}}_i\right)\ge |H|.
\]
Thus the quotient subspaces satisfy the same condition. By the induction hypothesis, we can choose linearly independent vectors $\overline \bw_i\in \overline{\mathscr{W}}_i, \forall i\in[t]\setminus J.$
For each such $i$, choose a representative $\bw_i\in \mathscr{W}_i$ of $\overline \bw_i$. Then the vectors $\{\bw_i:i\in J\}$ together with the lifted vectors $\{\bw_i:i\in[t]\setminus J\}$ are linearly independent in $\mathscr W$.

It remains to consider the case in which no nonempty proper subset $J$ is tight. Then, for every nonempty proper subset $J\subsetneq[t]$,
\[
    \dim\left(\sum_{i\in J}\mathscr{W}_i\right)\ge |J|+1.
\]
Since the condition applied to $\{t\}$ gives $\dim \mathscr{W}_t\ge1$, choose a nonzero vector $\bw_t\in \mathscr{W}_t.$
Consider the quotient space $\mathscr{W}/\langle \bw_t\rangle$. For $i\in[t-1]$, define
\[
    \overline{\mathscr{W}}_i\triangleq (\mathscr{W}_i+\langle \bw_t\rangle)/\langle \bw_t\rangle.
\]
For any $J\subseteq[t-1]$, let $\mathscr{W}_J=\sum_{i\in J}W_i$. If $J=\emptyset$, the desired inequality is trivial. If $J\neq\emptyset$, then
\[
    \dim \mathscr{W}_J\ge |J|+1.
\]
If $\bw_t\in \mathscr{W}_J$, then
\[
    \dim\left((\mathscr{W}_J+\langle \bw_t\rangle)/\langle \bw_t\rangle\right)
    =
    \dim \mathscr{W}_J-1
    \ge |J|.
\]
If $\bw_t\notin \mathscr{W}_J$, then
\[
    \dim\left((\mathscr{W}_J+\langle \bw_t\rangle)/\langle \bw_t\rangle\right)
    =
    \dim \mathscr{W}_J
    \ge |J|+1
    \ge |J|.
\]
Therefore the quotient subspaces $\overline{\mathscr{W}}_1,\ldots,\overline{\mathscr{W}}_{t-1}$ satisfy the required dimension condition. By induction, choose linearly independent vectors $\overline \bw_i\in\overline{\mathscr{W}}_i,$ $\forall i\in[t-1]$.
Lift each $\overline \bw_i$ to a vector $\bw_i\in \mathscr{W}_i$. Since the images of $\bw_1,\ldots,\bw_{t-1}$ are linearly independent in $\mathscr{W}/\langle \bw_t\rangle$, the vectors $\bw_1,\bw_2,\ldots,\bw_{t-1},\bw_t$ are linearly independent in $\mathscr{W}$. This completes the induction.
\end{pf}

The next lemma allows each local space to contribute several vectors.

\begin{lem}\label{lem:capacitated-covering}
Let $\mathscr{W}$ be a $d$-dimensional vector space in $\F_q^N$ with $d\leq N$, and let $\mathscr{W}_1,\mathscr{W}_2,\ldots,\mathscr{W}_m\subseteq \mathscr{W}$ be subspaces. Let $u_1,u_2,\ldots,u_m$ be nonnegative integers. For $B\subseteq[m]$, define
\[
    \rho(B)\triangleq
    \dim\left(\sum_{i\in B}\mathscr{W}_i\right).
\]
Then there exist $m$ subspaces $\mathscr{U}_i\subseteq \mathscr{W}_i$ with $\dim (\mathscr{U}_i)\le u_i, i\in[m]$, such that $\mathscr{U}_1+\mathscr{U}_2+\cdots+\mathscr{U}_m=\mathscr{W}$ if and only if for every $B\subseteq[m]$, 
\begin{equation}\label{eq:capacitated-covering}
    \rho(B)+\sum_{i\notin B}u_i\ge d.
\end{equation}
\end{lem}
\begin{pf}
For necessity, fix \(B\subseteq[m]\). The spaces \(\mathscr U_i\), \(i\in B\), contribute at most \(\rho(B)\) dimensions, while the remaining spaces contribute at most \(\sum_{i\notin B}u_i\) dimensions. Their sum cannot be all of \(\mathscr W\) unless \eqref{eq:capacitated-covering} holds.

For sufficiency, assume that for every $B\subseteq[m]$, \eqref{eq:capacitated-covering} holds. Let
\[
    u_{\rm tot}\triangleq\sum_{i=1}^m u_i.
\] The condition for \(B=\emptyset\) gives \(u_{\rm tot}\ge d\), while the condition for \(B=[m]\) gives \(\sum_i\mathscr W_i=\mathscr W\). Let \(\mathscr W'\) be an auxiliary space of dimension \(u_{\rm tot}-d\), and define \(\widehat {\mathscr W}=\mathscr W\oplus\mathscr  W'\).

For each $i\in[m]$, make $u_i$ labeled copies of $\mathscr{W}_i$. Denote these copies by $\mathscr{W}_{i,a}=\mathscr{W}_i,a\in[u_i].$ For every copy, define the extended subspace
\[
    \widehat{\mathscr{W}}_{i,a}\triangleq \mathscr{W}_i\oplus \mathscr{W}'\subseteq \widehat{\mathscr{W}}.
\]
There are $u_{\rm tot}$ extended subspaces in total.

Consider any collection $J\subseteq\{(i,a):i\in[m],a\in[u_i]\}$ of these copies, and let \(B(J)\) be the set of indices represented in \(J\). If \(J\neq\emptyset\), then
\[
    \sum_{(i,a)\in J}\widehat{\mathscr{W}}_{i,a}
    =
    \left(\sum_{i\in B( J)}\mathscr{W}_i\right)\oplus \mathscr{W}'.
\]
Thus,
\[
    \dim\left(\sum_{(i,a)\in J}\widehat{\mathscr{W}}_{i,a}\right)
    =
    \rho(B(J))+u_{\rm tot}-d.
\]
By \eqref{eq:capacitated-covering},
\[
    \rho(B(J))+u_{\rm tot}-d
    \ge
    \sum_{i\in B(J)}u_i
    \ge |J|.
\]
Thus, by Lemma~\ref{lem:linear-rado}, we can choose one vector $\widehat \bw_{i,a}\in \widehat{\mathscr{W}}_{i,a}$ from each copy such that all chosen vectors are linearly independent in $\widehat{\mathscr{W}}$.

Since the number of chosen vectors is $u_{\rm tot}$ and $\dim \widehat{\mathscr{W}}
    =
    \dim \mathscr{W}+\dim \mathscr{W}'
    =
    d+(u_{\rm tot}-d)
    =
    u_{\rm tot},$
these chosen vectors form a basis of $\widehat{\mathscr{W}}$.
Write each chosen vector as $\widehat \bw_{i,a}=\bw_{i,a}+\bw'_{i,a},$ where $\bw_{i,a}\in \mathscr{W}_i,\bw'_{i,a}\in \mathscr{W}'.$
Let $\mathscr{P}\triangleq \Span\{\bw_{i,a}:i\in[m],\,a\in[u_i]\}\subseteq \mathscr{W}.$
All chosen vectors $\widehat \bw_{i,a}$ lie in $\mathscr{P}\oplus \mathscr{W}'$. Since they span
$\widehat{\mathscr{W}}$, we must have $\mathscr{P}\oplus \mathscr{W}'=\widehat{\mathscr{W}}.$
Thus,
\[
    \dim \mathscr{P}+\dim \mathscr{W}'\ge \dim \widehat{\mathscr{W}}=u_{\rm tot}.
\]
Using $\dim \mathscr{W}'=u_{\rm tot}-d$, we get $\dim \mathscr{P}\ge d.$
But $\mathscr{P}\subseteq \mathscr{W}$ and $\dim \mathscr{W}=d$, so $\mathscr{P}=\mathscr{W}$.
For each $i\in[m]$, define $\mathscr{U}_i\triangleq \Span\{\bw_{i,a}:a\in[u_i]\}.$
Then $\mathscr{U}_i\le W_i,\dim \mathscr{U}_i\le u_i,$
and the sum of the $\mathscr{U}_i$'s contains $\mathscr{P}=\mathscr{W}$. Hence
\[
    \mathscr{U}_1+\mathscr{U}_2+\cdots+\mathscr{U}_m=\mathscr{W}.
\]
This proves sufficiency.
\end{pf}

We now apply Lemma~\ref{lem:capacitated-covering} to the local sections and complete the proof of Theorem~\ref{thm:fixed_auxiliary_space}.

\begin{pf}
First suppose that such a matrix \(E\) exists.  Let \(\mathscr{U}_i\triangleq \Row(E_i)\), where $E_i$ is the submatrix of $E$ corresponding to the $i$-th middle node as defined in Definition~\ref{def:support-constrained-matrix}.   The support-constraint gives
\[
        \mathscr{U}_i\subseteq \mathscr{W}\cap \mathscr{E}_{\Sigma_i}^{(\ell)}=\mathscr{L}_i(\mathscr{W}),
        \text{ and } \dim \mathscr{U}_i\le n .
\]
Moreover, \(\sum_{i=1}^m \mathscr{U}_i=\Row(E)=\mathscr{W}\).  For any \(B\subseteq[m]\), the subspaces indexed by \(B\) can contribute at most \(\rho_\mathscr{W}(B)\) dimensions, and the remaining \(m-|B|\) blocks can contribute at most \((m-|B|)n\) dimensions. Thus \eqref{eq:fixed-auxiliary-row} is necessary.

Conversely, assume that \eqref{eq:fixed-auxiliary-row} holds for every \(B\subseteq[m]\).  Apply the
capacitated subspace covering lemma to the ambient space \(\mathscr{W}\), the subspaces
\(\mathscr{L}_i(\mathscr{W})\), and the uniform capacities \(u_i=n\).  We obtain subspaces
\(\mathscr{U}_i\subseteq \mathscr{L}_i(\mathscr{W})\) with \(\dim \mathscr{U}_i\le n\) such that
\[
        \mathscr{U}_1+\mathscr{U}_2+\cdots+\mathscr{U}_m=\mathscr{W} .
\]
Choose a basis of each \(\mathscr{U}_i\), pad it by zero rows to obtain an \(n\times s\ell\) matrix \(E_i\), and stack these blocks to form \(E\). Then \(E\) is \(\Sigma\)-support-constrained and \(\Row(E)=\mathscr{W}\).  Since \(\mathscr{C}_T^{(\ell)}\subseteq \mathscr{W}=\Row(E)\), the support-constrained received-matrix characterization gives a valid linear \((\ell,n)\) computation code.
\end{pf}

The case \(B=[m]\) in \eqref{eq:fixed-auxiliary-row} has a special role. It requires
\[\rho_{\mathscr W}([m])=d_{\mathscr W},\]
which means that the local sections together generate \(W\). This condition does not depend on \(n\). If it fails, the fixed space \(W\) cannot be realized at any load. Once it holds, the remaining inequalities determine the minimum load exactly.

\begin{cor}\label{cor:fixed-space-rate}
    For a fixed instance length $\ell$ and a subspace $\mathscr{W}$, the following statements hold:
    \begin{enumerate}
        \item If $ \rho_\mathscr{W}([m])<d_\mathscr{W},$ the maximum achievable rate $R_\ell(\mathscr{W})=0$;
        \item If $ \rho_\mathscr{W}([m])=d_\mathscr{W},$ then the maximum achievable rate
        \[R_\ell(\mathscr{W})=\frac{\ell}{n^\star(\mathscr{W})},\]
        where $n^\star(\mathscr{W})$ is defined as \begin{equation}\label{eq:requiredl-load}
    n^\star(\mathscr{W})
        =
        \left\lceil
        \max_{B\subsetneq[m]}
        \frac{d_\mathscr{W}-\rho_\mathscr{W}(B)}{m-|B|}
        \right\rceil .  
\end{equation}
    \end{enumerate}
\end{cor}
\begin{pf}
    When $B=[m]$, the inequality \eqref{eq:fixed-auxiliary-row} implies $ \rho_\mathscr{W}([m])=d_\mathscr{W},$ which says that \(\mathscr{W}\) must be generated by its locally supported sections.  If this global condition fails, no finite value of \(n\) can realize the fixed row space \(\mathscr{W}\).  If it holds, the minimum required per-middle-node load is $n^\star(\mathscr{W})$. Hence the largest rate certified by this fixed row space is $R_\ell(\mathscr{W})=\ell /n^\star(\mathscr{W}).$
\end{pf}

The outer optimization over \(\mathscr W\) now gives the linear capacity.

\begin{cor}[Linear-capacity characterization]
\label{cor:linear-capacity-row-spaces}
The linear computing capacity is
\begin{equation}
    C_{\rm lin}(\cN_{s,m,\Sigma},T)
    =
    \sup_{\ell\ge1}
    \max_{\substack{\mathscr W:
        \mathscr C_T^{(\ell)}\subseteq \mathscr W\subseteq\F_q^{s\ell}\\
        \rho_{\mathscr W}([m])=\dim\mathscr W}}
    \frac{\ell}{n^\star(\mathscr W)},
    \label{eq:linear-capacity-row-space-form}
\end{equation}
where \(n^\star(\mathscr W)\) is given by \eqref{eq:requiredl-load}.
\end{cor}

The inner maximum in \eqref{eq:linear-capacity-row-space-form} is understood as zero if no locally generated space \(W\) exists for that \(\ell\).
This characterization separates two tasks. The inner problem, now solved, computes the best load for a prescribed \(\mathscr W\). The outer problem chooses a space \(\mathscr W\) that contains the target rows and has rich local sections. The rest of the paper uses this separation: direct constructions take the smallest possible \(\mathscr W\), while auxiliary constructions enlarge it in a controlled way.

\subsection{The direct target row space}

The smallest admissible global space is $\mathscr W=\mathscr C_T^{(\ell)}=\Row(I_\ell\otimes T)$. We first define the relevant local ranks for one function instance. Let $\mathscr C_T\triangleq\Row(T)\subseteq\F_q^s$ and $\mathscr L_i(T)\triangleq\mathscr C_T\cap\mathscr E_{\Sigma_i}^{(1)}$. For \(B\subseteq[m]\), let
\begin{equation}
    \rho_T(B)
    \triangleq
    \dim\left(\sum_{i\in B}\mathscr L_i(T)\right).
    \label{eq:direct-local-rank}
\end{equation}
Because both the target space and the support spaces split across the \(\ell\) instances,
\begin{equation}
    \mathscr L_i(\mathscr C_T^{(\ell)})
    =\bigoplus_{a=1}^{\ell}\mathscr L_i(T),
    \qquad
    \rho_{\mathscr C_T^{(\ell)}}(B)=\ell\rho_T(B).
    \label{eq:lifted-direct-local-rank}
\end{equation}
This identity is the reason that the direct rate can be stated entirely in terms of the single-instance spaces \(\mathscr L_i(T)\).

\begin{dfn}[Direct target-row construction]\label{def:direct-realization}
A support-constrained realization is called \emph{direct} if the received matrix
$E$ can be chosen so that
\[
    \Row(E)=\mathscr C_T^{(\ell)}=\Row(I_\ell\otimes T).
\]
Equivalently, every transmitted row can be selected from the lifted target row
space $\mathscr{C}_T^{(\ell)}=\Row(I_\ell\otimes T).$
\end{dfn}

Theorem~\ref{thm:fixed_auxiliary_space} now gives the exact direct rate.

\begin{thm}\label{thm:direct-rate} For positive integers \(\ell\) and \(n\), a direct target-row \((\ell,n)\) construction exists if and only if for every $B\subseteq[m]$, 
\begin{equation}\label{eq:direct-rate-condition}
    \ell\rho_T(B)+(m-|B|)n\ge k\ell.
\end{equation}
Consequently, the largest rate achievable by direct target-row constructions is 
\begin{equation}\label{eq:direct-rate}
    R_{\rm direct}(T,\Sigma) = \min_{\substack{B\subseteq[m]\\ \rho_T(B)<k}} \frac{m-|B|}{k-\rho_T(B)}.
\end{equation}
\end{thm}

\begin{pf}
For \(\mathscr W=\mathscr C_T^{(\ell)}\), we have \(d_{\mathscr W}=k\ell\). Substituting \eqref{eq:lifted-direct-local-rank} into Theorem~\ref{thm:fixed_auxiliary_space} gives \eqref{eq:direct-rate-condition}.

For every \(B\) with \(\rho_T(B)<k\), this condition is equivalent to
\[\frac{\ell}{n}\le\frac{m-|B|}{k-\rho_T(B)}.\]
Hence, every direct construction has rate at most the right-hand side of \eqref{eq:direct-rate}. If this minimum is zero, then \(\rho_T([m])<k\), so the local sections do not generate the target space and the direct rate is zero. Otherwise, the minimum is a positive rational number. Choosing \(\ell\) and \(n\) so that \(\ell/n\) equals this number makes all inequalities in \eqref{eq:direct-rate-condition} hold. This proves \eqref{eq:direct-rate}.
\end{pf}

The formula shows exactly why direct rows may fail. A set \(B\) with a small rank supply forces the nodes outside \(B\) to carry the missing target-space directions. We next compare this requirement with the capacity upper bounds in Section~\ref{sec:prelim}.

Let
\[
    \overline{C}_{\rm cut}(T,\Sigma)
    \triangleq
    \min_{\emptyset\neq A\subseteq[s]}
    \frac{|\Gamma(A)|}{\Rank(T_A)}
\]
be the simple cut set upper bound, and let
\[
    \overline{C}_{\rm sp}(T,\Sigma)
    \triangleq
    \min_{\substack{B\subseteq[m]\\ I_B\neq\emptyset}}
    \frac{|B|}
    {\max_{\cP_B\in\mathcal P(B)} R_T(\cP_B)}
\]
be the strong-partition upper bound. In the following, let \(\overline{C}\) denote either one of these bounds. Since the source subset \(A=[s]\) gives
\[
    \overline{C}_{\rm cut}(T,\Sigma)
    \le
    \frac{|\Gamma([s])|}{\Rank(T)}
    \le
    \frac{m}{k},
\]
we also have \(\overline{C}\le m/k\). 
Note that the quantity
\[\lambda_{\overline C}\triangleq m-k\overline C\]
is nonnegative. We call it the \emph{cut slack}. 
At rate \(\overline{C}\), the set \(B\) must supply enough target-row dimensions so that the remaining \(m-|B|\) middle-to-sink edges can provide the rest. In particular, if \(|B|>\lambda_{\overline{C}}\), the required rank contribution of \(B\) is
\[
    \frac{|B|-\lambda_{\overline{C}}}{\overline{C}}.
\]
We call this quantity the \emph{cut-adjusted rank demand} of \(B\). If \(|B|\le \lambda_{\overline{C}}\), the demand is nonpositive and hence no rank requirement is imposed on \(B\). Thus the natural optimality condition is that every set of middle nodes has rank supply at least its cut-adjusted rank demand.

\begin{thm}[Cut-adjusted rank-demand criterion]
\label{thm:cut-adjusted-rank-demand}
Let \(\overline C>0\) be either \(\overline C_{\rm cut}(T,\Sigma)\) or \(\overline C_{\rm sp}(T,\Sigma)\). The direct target-row construction achieves \(\overline C\) if and only if $\forall B\subseteq[m]\text{ with } |B|>\lambda_{\overline{C}}$, 
\begin{equation}
    \rho_T(B)
    \ge
    \frac{|B|-\lambda_{\overline{C}}}{\overline{C}}.
    \label{eq:rank-demand-condition}
\end{equation}
\end{thm}
\begin{pf}
Because \(\overline C\) is an upper bound and \(R_{\rm direct}\) is achievable, equality holds if and only if \(R_{\rm direct}\ge\overline C\). By \eqref{eq:direct-rate}, this is equivalent to
\[
    \frac{m-|B|}{k-\rho_T(B)}\ge\overline C
\]
for every \(B\) with \(\rho_T(B)<k\). Rearranging gives
\[
    \rho_T(B)
    \ge
    k-\frac{m-|B|}{\overline C}
    =
    \frac{|B|-\lambda_{\overline C}}{\overline C}.
\]
When \(|B|\le\lambda_{\overline C}\), the right-hand side is nonpositive, so the inequality holds automatically. Sets with \(\rho_T(B)=k\) also satisfy it automatically because \(\overline C\le m/k\). Hence it is enough, and necessary, to check the sets stated in the theorem.
\end{pf}

The criterion has a simple interpretation: after the cut slack is removed, the rank supplied by \(B\) must grow at least at rate \(1/\overline C\) with \(|B|\). In the special case \(\overline C=m/k\), the slack is zero and the condition reduces to
\[
    \rho_T(B)\ge\frac{k}{m}|B|,
    \qquad B\subseteq[m].
\]
Thus every set of middle nodes must supply at least its proportional share of the \(k\)-dimensional target space.

The following basis-block condition gives a direct way to verify this balanced case. A set \(S\subseteq[m]\) is called a \emph{basis block} if \(|S|=k\) and one can choose \(\bc_i\in \mathscr L_i(T)\), \(i\in S\), such that \(\{\bc_i:i\in S\}\) is a basis of \(\mathscr C_T\).

\begin{thm}
\label{thm:uniform-cover-basis-blocks}
Suppose there is a multiset of basis blocks $\mathcal S=\{S_1,S_2,\ldots,S_\ell\}$ such that every middle node $i\in[m]$ appears in exactly $n$ blocks of $\mathcal S$. Then there exists a direct target-row $(\ell,n)$ construction. In particular,
\[R_{\rm direct}(T,\Sigma)=\frac{m}{k}.\]
Consequently, the direct construction attains the global cut set bound.
\end{thm}
\begin{pf}
    For each \(h\in[\ell]\), choose vectors \(\bc_i^{(h)}\in \mathscr L_i(T)\), \(i\in S_h\), that form a basis of \(\mathscr C_T\). For the \(h\)-th instance, node \(v_i\) sends \(\bc_i^{(h)}\bx^{(h)}\) when \(i\in S_h\). This is locally computable because \(\bc_i^{(h)}\) is supported on \(\Sigma_i\). The \(k\) symbols associated with \(S_h\) determine all \(k\) target values for that instance because their coefficient rows form a basis of \(\Row(T)\).
    
    Every node appears in \(n\) blocks and therefore sends \(n\) symbols. Also, counting block--node incidences gives \(k\ell=mn\), so the achieved rate is \(\ell/n=m/k\). The simple cut-set bound is at most \(m/k\); hence the construction is optimal.
\end{pf}

We next specialize the direct-rate formula to three basic target classes. These cases also explain when auxiliary rows are necessary.

\paragraph{Identity targets}

The identity target $T=I_s$ corresponds to full recovery of all source messages. Although this is not a genuine computation task, it is an important benchmark because it reduces the three-layer computation problem to a classical reconstruction problem. More importantly, it clarifies the meaning of direct target-row constructions.

For \(T=I_s\), the target space is the whole ambient space. Hence,
\begin{equation*}
   \mathscr L_i=\mathscr{C}_T\cap\mathscr E_{\Sigma_i}=\mathscr E_{\Sigma_i},
\end{equation*}
and for every $B\subseteq[m]$,
\begin{equation*}
    \rho_T(B)
    =
    \dim\left(\sum_{i\in B}\mathscr E_{\Sigma_i}\right)
    =
    \left|\bigcup_{i\in B}\Sigma_i\right|.
\end{equation*}
Every supported row is already a target row, so the direct restriction loses nothing. 
The following result shows that the exact direct rate reduces to the usual simple cut/rank bound. 

\begin{cor}[Identity targets]
For $T=I_s$,
\begin{equation*}
    R_{\rm direct}(I_s,\Sigma)
    =
    \overline C_{\rm cut}(I_s,\Sigma)
    =
    \min_{\emptyset\neq A\subseteq[s]}
    \frac{|\Gamma(A)|}{|A|}.
\end{equation*}
\end{cor}

\begin{pf}
For $T=I_s$, the direct rate formula gives 
\begin{equation}\label{eq:identity-direct-rate-middle-form}
    R_{\rm direct}(I_s,\Sigma)
    =
    \min_{\substack{B\subseteq[m]\\
    \cup_{i\in B}\Sigma_i\neq [s]}}
    \frac{m-|B|}
    {s-\left|\cup_{i\in B}\Sigma_i\right|}.
\end{equation}
We prove that the right-hand side of \eqref{eq:identity-direct-rate-middle-form} is equal to the cut set bound $\min_{\emptyset\neq A\subseteq[s]}
    \frac{|\Gamma(A)|}{|A|}$.

First fix $B\subseteq[m]$ such that $\cup_{i\in B}\Sigma_i\neq[s]$, and define $A_B\triangleq [s]\setminus\bigcup_{i\in B}\Sigma_i.$
Then $A_B\neq\emptyset$. Moreover, no middle node in $B$ observes any source in $A_B$, and hence $\Gamma(A_B)\subseteq [m]\setminus B.$
Therefore,
\begin{equation*}
    \frac{|\Gamma(A_B)|}{|A_B|}
    \le
    \frac{m-|B|}
    {s-\left|\cup_{i\in B}\Sigma_i\right|}.
\end{equation*}
Taking the minimum over $B$ gives
$\overline C_{\rm cut}(I_s,\Sigma)
    \le
    R_{\rm direct}(I_s,\Sigma).$

Conversely, fix a nonempty source subset $A\subseteq[s]$, and define $ B_A\triangleq [m]\setminus\Gamma(A).$
Then every middle node in $B_A$ is disjoint from $A$, so $ \bigcup_{i\in B_A}\Sigma_i\subseteq [s]\setminus A.$
Thus,
\begin{equation*}
    s-\left|\bigcup_{i\in B_A}\Sigma_i\right|
    \ge
    |A|.
\end{equation*}
Also, $m-|B_A|=|\Gamma(A)|$. Hence
\begin{equation}
    R_{\rm direct}(I_s,\Sigma)
    \le
    \frac{m-|B_A|}
    {s-\left|\cup_{i\in B_A}\Sigma_i\right|}
    \le
    \frac{|\Gamma(A)|}{|A|}.
\end{equation}
Taking the minimum over all nonempty $A\subseteq[s]$ gives
\begin{equation}
    R_{\rm direct}(I_s,\Sigma)
    \le
    \overline C_{\rm cut}(I_s,\Sigma).
\end{equation}
Combining the two inequalities proves the claim.
\end{pf}

\paragraph{Algebraic sum targets}
It is useful to contrast the identity target with the scalar sum target
\(T=\mathbf 1_s^T\). In this case \(\mathscr{C}_T=\Span\{\mathbf 1_s\}\). Hence
\[
    \mathscr L_i=\mathscr{C}_T\cap\mathscr E_{\Sigma_i}
    =
    \begin{cases}
        \mathscr{C}_T, & \Sigma_i=[s],\\
        \{0\}, & \Sigma_i\neq[s].
    \end{cases}
\]
Thus, a nonzero direct row is available only at a node that observes every source. This does not mean that sum computation itself is difficult. Instead, it shows that ordinary sum computation generally relies on rows outside the target row space. For instance, if \(s=m=2\), \(\Sigma_1=\{1\}\), and \(\Sigma_2=\{2\}\), then \(v_1\) can send \(x_1\) and \(v_2\) can send \(x_2\), and the sink computes \(x_1+x_2\). The received rows \((1,0)\) and \((0,1)\) are not target rows, but their span contains the target row \((1,1)\). The scalar sum therefore gives the simplest example in which an auxiliary global space is useful.

\paragraph{Groupwise sum targets}

Suppose that \([s]=G_1\sqcup G_2\sqcup\cdots\sqcup G_k\), and that the sink wants one sum from each group. Then
\[
    \mathscr C_T
    =
    \Span\{\mathbf 1_{G_1},\ldots,\mathbf 1_{G_k}\}.
\]
A vector in \(\mathscr C_T\) is constant on every group on which it is nonzero. Hence node \(v_i\) can directly use the row \(\mathbf 1_{G_a}\) exactly when it observes the whole group \(G_a\). Define
\[
    Q_i
    \triangleq
    \{a\in[k]:G_a\subseteq\Sigma_i\}.
\]
Then,
\begin{equation}
    \mathscr L_i(T)
    =
    \Span\{\mathbf 1_{G_a}:a\in Q_i\},
    \qquad
    \rho_T(B)
    =
    \left|\bigcup_{i\in B}Q_i\right|.
    \label{eq:groupwise-local-space}
\end{equation}
Consequently,
\begin{equation}
    R_{\rm direct}(T,\Sigma)
    =
    \min_{\substack{B\subseteq[m]\\\cup_{i\in B}Q_i\neq[k]}}
    \frac{m-|B|}
    {k-\left|\cup_{i\in B}Q_i\right|}.
    \label{eq:groupwise-direct-rate}
\end{equation}
Thus the direct problem depends only on the derived group-level access sets \(Q_1,\ldots,Q_m\). It is exactly the identity-recovery problem for the \(k\) group sums. In particular, if every \(\Sigma_i\) is a union of whole groups, then the original computation problem itself reduces to this group-level identity problem, and direct rows achieve its simple cut-set bound.

Partial group observations are different. They may help compute the group
sums, but they produce no direct target row. For example, let
\[
    T=
    \begin{bmatrix}
        1&1&0&0\\
        0&0&1&1
    \end{bmatrix},
    \qquad
    \Sigma_i=\{i\},\quad i\in[4].
\]
Here every \(Q_i\) is empty, so the direct rate is zero. Nevertheless, the
four nodes can send \(x_1,x_2,x_3,x_4\), and the sink can recover both group
sums at rate one. The coordinate rows enlarge the target space and give an
auxiliary realization.

\subsection{Choosing auxiliary row spaces}
The direct space is only one candidate in \eqref{eq:linear-capacity-row-space-form}. We now give a rank-based rule for deciding whether an added direction improves a current choice \(\mathscr W\).

For \(\mathscr W\supseteq\mathscr C_T^{(\ell)}\), define the \emph{global uncovered
dimension}
\[
    g(\mathscr W)
    \triangleq
    d_{\mathscr W}-\rho_{\mathscr W}([m]).
\]
As illustrated in Section~\ref{subsec:fixed_auxiliary_rows}, \(g(\mathscr W)=0\) is exactly the requirement that \(\mathscr W\) be generated by its locally supported sections.
If \(g(\mathscr W)>0\), then the fixed row space \(W\) cannot be realized for any finite value of \(n\).
For every proper \(B\subsetneq[m]\), define
\[
    \theta_B(\mathscr W)
        \triangleq
        \frac{d_{\mathscr W}-\rho_{\mathscr W}(B)}{m-|B|},
    \qquad
    \Phi(\mathscr W)\triangleq \max_{B\subsetneq[m]}\theta_B(\mathscr W).
\]
When \(g(W)=0\), Corollary~\ref{cor:fixed-space-rate} gives
\[n^\star(W)=\lceil\Phi(W)\rceil.\]
Thus \(g(\mathscr W)\) tests whether \(W\) is realizable at all, while \(\Phi(\mathscr W)\) measures the load once this global obstruction has been removed.

%If each middle node is allowed to transmit at most \(n\) rows, then the fixed-\(W\) criterion in
%Theorem~\ref{thm:fixed_auxiliary_space} can be written as
%\[
%        \rho_{\mathscr W}(B)\ge d_{\mathscr W}-(m-|B|)n.
%\]
%Denote the  rank demand of \(B\) as $\mathsf D_n(B;W)\triangleq d_{\mathscr W}-(m-|B|)n$,
%and denote the corresponding rank deficit as
%\[
%        \delta_n(B;W)
%        \mathsf D_n(B;W)-\rho_{\mathscr W}(B)
%        =
%        d_{\mathscr W}-\rho_{\mathscr W}(B)-(m-|B|)n .      
%\]
%The fixed row space \(W\) is realizable with load \(n\) if and only if
%\[
%        \delta_n(B;W)\le 0,\qquad \forall B\subseteq[m].      
%\]
%In particular, for \(B=[m]\), this deficit becomes $\delta_n([m];W)=d_{\mathscr W}-\rho_{\mathscr W}([m]).$
%This quantity is independent of \(n\).  We call it the global uncovered
%dimension and write
%\[
%        g(W)\triangleq d_{\mathscr W}-\rho_{\mathscr W}([m]).   
%\]
%As illustrated in Section~\ref{subsec:fixed_auxiliary_rows},
%\(g(W)=0\) is exactly the requirement that \(W\) be generated by
%its locally supported sections.  

Now choose a direction \(\ba\notin \mathscr W\), and let
\[
        \mathscr W^+(\ba)\triangleq \mathscr W+\langle \ba\rangle .
\]
For \(B\subseteq[m]\), define the local-rank gain
\[
        \Delta_B(\ba;\mathscr W)
        \triangleq
        \rho_{\mathscr W^+(\ba)}(B)-\rho_{\mathscr W}(B).       
\]
Since \(\ba\notin \mathscr W\), we have \(\dim \mathscr W^+(\ba)=d_{\mathscr W}+1\).  Therefore, for every
\(B\subseteq[m]\),
\begin{equation}\label{eq:basic-accounting1}
    \theta_B(\mathscr W^+(\ba))
        =
        \theta_B(\mathscr W)
        +
        \frac{1-\Delta_B(\ba;\mathscr W)}{m-|B|}.
\end{equation}
For the global obstruction,
\begin{equation}\label{eq:basic-accounting2}
    g(\mathscr W^+(\ba))
        =
        g(\mathscr W)+1-\Delta_{[m]}(\ba;\mathscr W).
\end{equation}

Equations \eqref{eq:basic-accounting1} and \eqref{eq:basic-accounting2} show the two effects of an auxiliary direction. It raises \(\dim\mathscr W\) by one, but it may unlock several new locally supported directions. For a particular set \(B\), its normalized deficit decreases exactly when \(\Delta_B(\ba;\mathscr W)>1\). Thus a good auxiliary row is not merely a new row, it should create several useful local directions for the current bottleneck sets.

This gives the following one-step greedy rule.  At the current row space \(\mathscr W\), choose an auxiliary direction \(\ba\notin \mathscr W\) that minimizes
\begin{equation}\label{eq:priority-rule}
    \left(
        g(\mathscr W^+(\ba)),\;
        \Phi(\mathscr W^+(\ba))
        \right) 
\end{equation}
in lexicographic order.  Equivalently, the first priority is to make the enlarged row space locally generated, i.e., to reduce the global uncovered dimension \(g(\mathscr W)\).  Once the global obstruction is removed, the second priority is to reduce the largest normalized rank deficit \(\Phi(\mathscr W)\), and hence the load is certified by Theorem~\ref{thm:fixed_auxiliary_space}.

The rule in \eqref{eq:priority-rule} is only a one-step criterion.  It does not claim to solve the global problem of selecting a multi-dimensional auxiliary space.  Nevertheless, it gives a systematic rank-testable principle: a useful auxiliary row should be chosen so that it increases the rank supply of the currently deficient middle-node subsets by more than the one-dimensional demand that it adds.

\begin{figure}
    \subfigure[An example of direct target-row construction.]{\begin{minipage}[b]{0.49\linewidth}\includegraphics[width=1\linewidth]{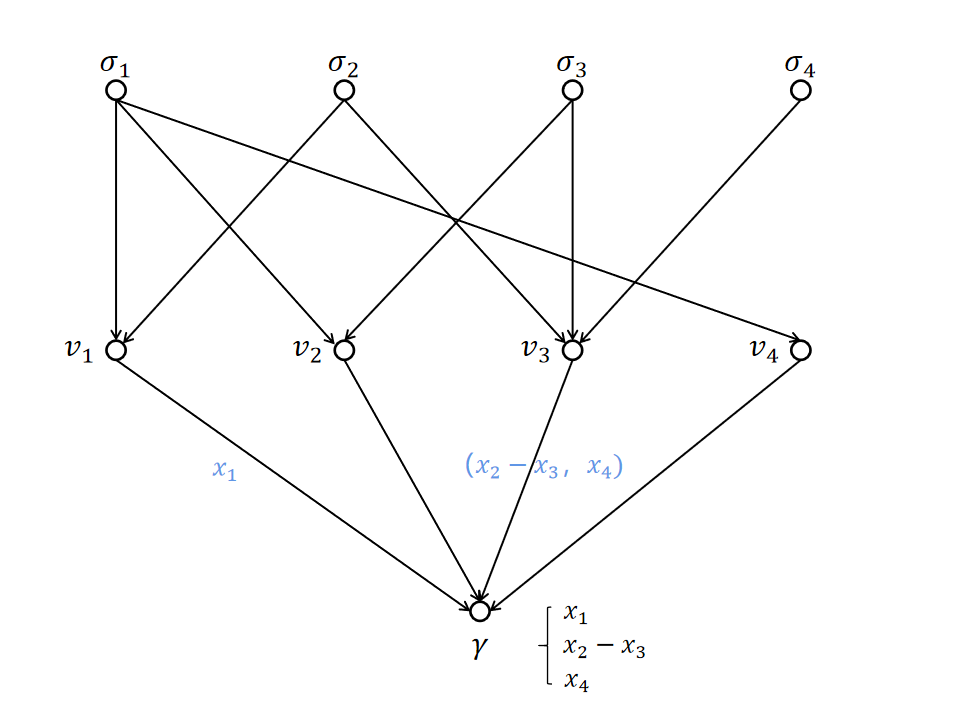}\end{minipage}}
    \vspace{4pt}
    \subfigure[Adding an auxiliary row.]{\begin{minipage}[b]{0.49\linewidth}\includegraphics[width=1\linewidth]{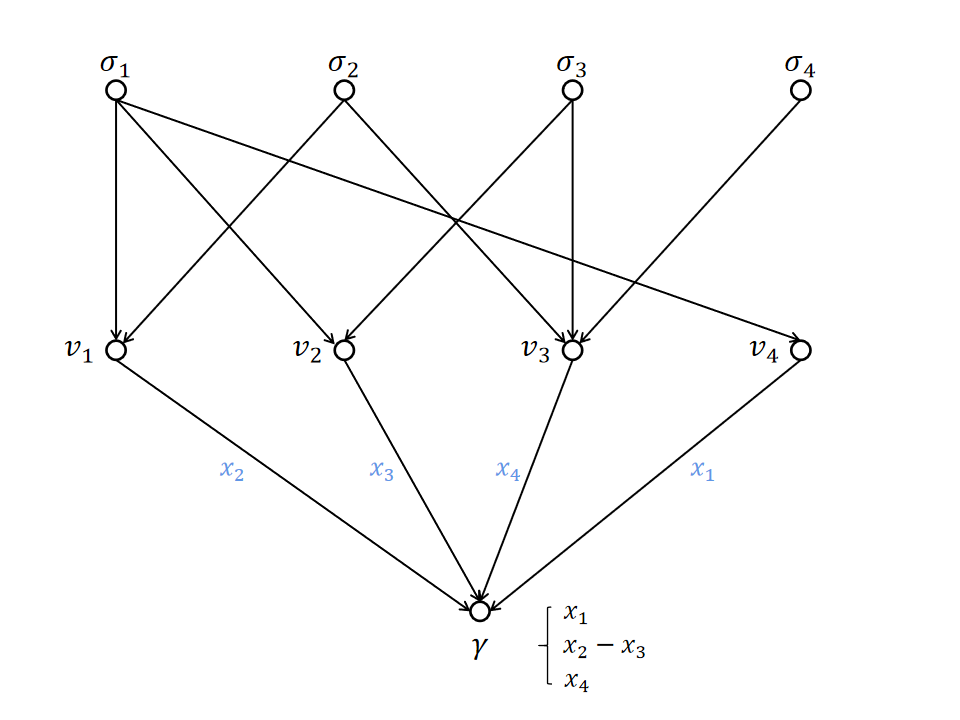}\end{minipage}}
    \caption{Adding an auxiliary row improves the computing capacity.}
    \label{fig:add_row}
\end{figure}

\begin{example}\label{ex:one_row_bottleneck}
Consider the network depicted in Fig.~\ref{fig:add_row}. Let \(s=m=4\) and \(\ell=1\).  Consider the target row space
\[
        \mathscr{C}_T
        =
        \Span\{\be_1,\ \be_2-\be_3,\ \be_4\}
        \subseteq \mathbb F_q^4,
\]
or equivalently,
\[
        T=
        \begin{bmatrix}
        1&0&0&0\\
        0&1&-1&0\\
        0&0&0&1
        \end{bmatrix}.
\]
Let the source-access sets be
\[
        \Sigma_1=\{1,2\},\qquad
        \Sigma_2=\{1,3\},\qquad
        \Sigma_3=\{2,3,4\},\qquad
        \Sigma_4=\{1\}.
\]
First take the direct row space $\mathscr W=\mathscr{C}_T .$
The local sections are
\[
        \mathscr L_1(\mathscr W)=\Span\{\be_1\},\qquad
        \mathscr L_2(\mathscr W)=\Span\{\be_1\}, \qquad
        \mathscr L_3(\mathscr W)=\Span\{\be_2-\be_3,\ \be_4\},\qquad
        \mathscr L_4(\mathscr W)=\Span\{\be_1\}.
\]
They generate \(\mathscr W\), so \(g(\mathscr W)=0\). However, for \(B=\{1,2,4\}\),
\[
        \rho_{\mathscr W}(B)
        =
        \dim\Span\{\be_1\}
        =
        1,\qquad
        \theta_B(\mathscr W)
        =
        \frac{d_{\mathscr W}-\rho_{\mathscr W}(B)}{m-|B|}
        =
        \frac{3-1}{1}
        =
        2.
\]
The nodes in \(B\) provide only \(e_1\), so node \(v_3\) must supply both remaining target directions. Hence \(n^\star(W)=2\).
In words, if we insist on using only rows from the target row space \(\mathscr{C}_T\), then the nodes \(v_1,v_2,v_4\) collectively provide only the direction \(\be_1\).  The two remaining target-row directions \(\be_2-\be_3\) and \(\be_4\) must both be supplied by \(v_3\), forcing load \(2\).

Now add one auxiliary direction $\ba=\be_2$ and define $\mathscr W^+(\ba)=\mathscr W+\langle \be_2\rangle $.
Since $\be_3=\be_2-(\be_2-\be_3)\in\mathscr  W^+(\ba),$
we have
\[
        \mathscr W^+(\ba)=\mathbb F_q^4.
\]
The new local sections are
\begin{align*}
    &\mathscr L_1(\mathscr W^+(\ba))=\Span\{\be_1,\be_2\},\qquad \mathscr L_2(\mathscr W^+(\ba))=\Span\{\be_1,\be_3\},\\
    &\mathscr L_3(\mathscr W^+(\ba))=\Span\{\be_2,\be_3,\be_4\},\qquad \mathscr L_4(\mathscr W^+(\ba))=\Span\{\be_1\}.
\end{align*}
For the same set \(B=\{1,2,4\}\), the local sections now supply
\[
        \rho_{\mathscr W^+(\ba)}(B)
        =
        \dim\Span\{\be_1,\be_2,\be_3\}
        =
        3.
\]
Thus, $\Delta_B(\ba;\mathscr W)=2.$
Although only one auxiliary row is added, it unlocks two new local directions for this bottleneck subset.  The rank demand also increases by one, since $\dim \mathscr W^+(\ba)=\dim \mathscr W+1.$
Hence the net deficit change is $1-\Delta_B(\ba;\mathscr W)=-1.$
Equivalently,
\[
        \theta_B(\mathscr W^+(\ba))
        =
        \frac{4-3}{1}
        =
        1.
\]
One checks that every proper subset \(B'\subsetneq[4]\) satisfies
\[
        \rho_{\mathscr W^+(\ba)}(B')+(4-|B'|)\ge 4.
\]
Therefore $ n^\star(\mathscr W^+(\ba))=1.$
The auxiliary row \(\be_2\) improves the certified rate from \(1/2\) to \(1\).
\end{example}
The example captures the main role of auxiliary rows. Enlarging \(\mathscr W\) is useful only when the added global dimension creates a larger gain in the local sections.

\section{MDS target under cyclic access}\label{sec:mds_target}
In this section, we focus on maximum distance separable (MDS) target functions and cyclic access systems. We first identify the common converse bound and the local threshold. We then give direct constructions in the dense regime. For the sparse regime, we construct a larger global row space through a sliding parity-check design. This design is deterministic, applies to every prescribed MDS target, and has a complete algebraic decodability proof.

A matrix $T\in\F_q^{k\times s}$ with $\Rank(T)=k$ is called an \emph{MDS target matrix} if for every $A\subseteq[s]$, $\Rank(T_A)=\min\{k,|A|\}$. MDS target matrices are a natural class of vector-linear targets for three reasons. First, over a sufficiently large field, they represent generic dense full-rank linear transformations. Second, for MDS targets, all rank terms appearing in cut-based converses depend only on the cardinalities of source subsets. Third, several familiar targets are boundary cases of MDS targets. In particular, the identity function, the algebraic sum, and the nondegenerate rank-$(s-1)$ target function have MDS coefficient matrices with $k=s,1$, and $s-1$, respectively.

\subsection{The MDS converse and the local threshold}
For an MDS target matrix, the simple cut set bound becomes the purely combinatorial quantity
\begin{equation}
    \overline{C}_{\rm MDS}(\Sigma)
    \triangleq
    \min_{\emptyset\neq A\subseteq[s]}
    \frac{|\Gamma(A)|}{\min\{k,|A|\}}.
    \label{eq:Cbar-MDS}
\end{equation}
The following lemma shows that, in the MDS case, the middle-layer specialization of the strong-partition bound does not improve this simple bound on three-layer source-access networks.

\begin{lem}\label{lem:mds-strong-equals-simple}
Assume that $T\in\F_q^{k\times s}$ is MDS and that every source is adjacent to at least one middle node. Then the middle-layer specialization of the strong-partition bound satisfies
\begin{equation}
    \min_{\substack{B\subseteq[m]\\ I_B\neq\emptyset}}
    \frac{|B|}
    {\max_{P_B\in\mathcal P(B)} R_T(P_B)}
    =
    \min_{\emptyset\neq A\subseteq[s]}
    \frac{|\Gamma(A)|}{\Rank(T_A)}
    =
    \overline{C}_{\rm MDS}(\Sigma).
    \label{eq:mds-two-bounds-equal}
\end{equation}
\end{lem}
\begin{pf}
Let $\alpha
    \triangleq
    \min_{\emptyset\neq A\subseteq[s]}
    \frac{|\Gamma(A)|}{\Rank(T_A)}.$
Since the trivial partition is always a strong partition, the strong-partition bound is no larger than the simple cut-set bound. Thus it remains to prove that, for every middle-layer cut $B\subseteq[m]$ and every strong partition $P_B=\{B_1,B_2,\ldots,B_t\}\in\mathcal P(B)$,
\begin{equation}
    |B|\ge \alpha R_T(P_B).
    \label{eq:proof-need-alpha}
\end{equation}

Fix such $B$ and $P_B$. 
Since $B_a\subseteq B$, we have $I_{B_a}\subseteq I_B$ for every $a$, and hence $\bigcup_{a=1}^t I_{B_a}\subseteq I_B$. Moreover, the strong-partition condition implies that the sets $I_{B_1},I_{B_2},\ldots,I_{B_t}$ are pairwise disjoint. 

If $|\bigcup_{a=1}^t I_{B_a}|\le k$, then, since the sets $I_{B_1},I_{B_2},\ldots,I_{B_t}$ are disjoint, we have
\begin{equation*}
    \Rank(T_{I_{P_B}})=\left|\bigcup_{a=1}^t I_{B_a}\right|=\sum_{a=1}^t |I_{B_a}|=\sum_{a=1}^t \Rank(T_{I_{B_a}}).
\end{equation*}
Therefore, from \eqref{eq:RTPC}, we obtain $R_T(P_B)=\Rank(T_{I_{B}})$. Since $\Gamma(I_{B})\subseteq B$, by the definition of $\alpha$ we get
\begin{equation*}
    |B|\ge |\Gamma(I_B)|\ge\alpha R_T(P_B).
\end{equation*}

Next, consider the case when $|\bigcup_{a=1}^t I_{B_a}|>k$. Then $\Rank(T_{I_{P_B}})=k$. Since $\bigcup_{a=1}^t I_{B_a}\subseteq I_{B}$, we also have $\Rank(T_{I_B})=k$. Hence, from \eqref{eq:RTPC}, we have $R_T(P_B)=\sum_{a=1}^t \Rank(T_{I_{B_a}}).$
For each $a\in[t]$, we have $\Gamma(I_{B_a})\subseteq B_a$. Thus
\begin{equation*}
    |B_a|\ge |\Gamma(I_{B_a})|\ge \alpha \Rank(T_{I_{B_a}}).
\end{equation*}
Summing over $a$ gives
\begin{equation*}
    |B|=\sum_{a=1}^t |B_a|
    \ge
    \alpha\sum_{a=1}^t \Rank(T_{I_{B_a}})
    =
    \alpha R_T(P_B).
\end{equation*}
This proves \eqref{eq:proof-need-alpha} in both cases. Therefore the strong-partition bound is at least $\alpha$, and hence it equals the simple cut/rank bound. Finally, since $T$ is MDS, $\Rank(T_A)=\min\{k,|A|\}$, so the common value is $\overline{C}_{\rm MDS}(\Sigma)$.
\end{pf}

\begin{rmk}\label{rem:complete-strong-partition-bound}
Lemma~\ref{lem:mds-strong-equals-simple} is stated in terms of the middle-layer specialization \eqref{eq:three-layer-strong-bound}. Under the MDS assumption, the same conclusion also holds for the complete strong-partition bound in Theorem~2, which permits arbitrary edge cuts. 
To see this, first observe that every edge cut set \(C\in\Lambda(\mathcal N_{s,m,\Sigma})\) satisfies $|C|\geq|\Gamma(I_C)|.$ Indeed, for every \(i\in\Gamma(I_C)\), choose a source \(j\in I_C\cap\Sigma_i\). The two-edge path $\sigma_j\rightarrow v_i\rightarrow\gamma$ must contain an edge of \(C\). The paths selected for distinct middle nodes contain disjoint edges, which implies $|C|\geq|\Gamma(I_C)|.$

Now let \(C\) be an arbitrary edge cut and let $\mathcal P_C=\{C_1,C_2,\ldots,C_t\}$ be an arbitrary strong partition of \(C\). As in the proof of Lemma~\ref{lem:mds-strong-equals-simple}, the sets \(I_{C_1},\ldots,I_{C_t}\) are pairwise disjoint. If $\left|\bigcup_{a=1}^{t}I_{C_a}\right|\leq k,$ then $R_T(\mathcal P_C)=\Rank(T_{I_C}).$
Consequently,
\[|C|\geq|\Gamma(I_C)|\geq\overline C_{\rm MDS}(\Sigma)R_T(\mathcal P_C).\]
If $\left|\bigcup_{a=1}^{t}I_{C_a}\right|>k,$ then, $R_T(\mathcal P_C)=\sum_{a=1}^{t}\Rank(T_{I_{C_a}}).$ Therefore, $|C_a|\geq|\Gamma(I_{C_a})|$  yields
\[|C|=\sum_{a=1}^{t}|C_a|\geq\overline C_{\rm MDS}(\Sigma)\sum_{a=1}^{t}\Rank(T_{I_{C_a}})=\overline C_{\rm MDS}(\Sigma)R_T(\mathcal P_C).\]
Therefore, every term in the complete strong-partition bound is at least \(\overline C_{\rm MDS}(\Sigma)\). On the other hand, the complete bound is no larger than its middle-layer specialization, because it minimizes over a larger collection of cuts. Consequently, under the assumptions of Lemma~\ref{lem:mds-strong-equals-simple}, the complete strong-partition bound, its middle-layer specialization, and the simple cut-set bound are all equal to \(\overline C_{\rm MDS}(\Sigma)\).
\end{rmk}

Lemma~\ref{lem:mds-strong-equals-simple} allows us to use $\overline C_{\rm MDS}(\Sigma)$ as the converse benchmark throughout the MDS part. We next determine when the target row space contains a nonzero vector with a prescribed support.
\begin{lem}[MDS shortening threshold]\label{lem:mds-shortening-threshold}
Let $T\in\F_q^{k\times s}$ be MDS. Then, for every $A\subseteq[s]$,
\begin{equation}
    \dim(\mathscr{C}_T\cap\mathscr E_A)=\bigl(|A|-s+k\bigr)^+,
    \label{eq:mds-shortening-formula}
\end{equation}
where $x^+=\max\{x,0\}$. 
In particular, $\mathscr L_i\neq\{0\}$ if and only if $|\Sigma_i|\ge s-k+1.$
\end{lem}

\begin{pf}
A vector in $\mathscr{C}_T$ has the form $\mathbf aT$ for some $\mathbf a\in\F_q^k$. It is supported in $A$ if and only if $\mathbf aT_{[s]\setminus A}=0.$
Since $T$ is MDS,
\begin{equation*}
    \Rank(T_{[s]\setminus A})=\min\{k,s-|A|\}.
\end{equation*}
Hence the dimension of the solution space for $\mathbf a$ is $k-\min\{k,s-|A|\}
    =
    \bigl(|A|-s+k\bigr)^+.$
Because $T$ has row rank $k$, the map $\mathbf a\mapsto\mathbf aT$ is injective. This proves \eqref{eq:mds-shortening-formula} and the remaining statement follows.
\end{pf}
The threshold in Lemma~\ref{lem:mds-shortening-threshold} is a local condition. Above it, each middle node has at least one nonzero target row that it can send. Direct optimality still depends on how these local rows fit together, as described by Theorem~\ref{thm:cut-adjusted-rank-demand}. Below the threshold, the situation is sharper, that is, every local section is zero, so no direct target-row realization is possible.

We now specialize to the cyclic support system.  Let $m=s$ and let all indices be cyclic in $[s]$.  For a fixed support size $r$, for every $i\in[s]$, define
\begin{equation}\label{eq:cyclic-Sigma}
        \Sigma_i=[i:i+r-1]_{\rm s}.
\end{equation}
Thus, source \(j\) is observed by
\begin{equation}
    \Gamma(j)=[j-r+1:j]_{\rm s}.
    \label{eq:cyclic-Gamma}
\end{equation}

\begin{lem}\label{lem:cyclic-bound}
    For the cyclic support system \eqref{eq:cyclic-Sigma} and an MDS target matrix $T\in\F_q^{k\times s}$,
\begin{equation}\label{eq:cyclic-CMDS}
        \overline C_{\rm MDS}(\Sigma)=\frac{\min\{r+k-1,s\}}{k}.
\end{equation}
\end{lem}
\begin{pf}
    For any nonempty $A\subseteq[s]$, the middle-node neighborhood $\Gamma(A)$ is the set of cyclic intervals of length $r$ that intersect $A$.  If $|A|=a$, then
\[
        |\Gamma(A)|\ge \min\{a+r-1,s\},
\]
with equality when $A$ is a cyclic interval of length $a$.  Therefore
\[
        \overline C_{\rm MDS}(\Sigma)
        =\min_{1\le a\le s}\frac{\min\{a+r-1,s\}}{\min\{a,k\}}.
\]
For \(1\leq a\leq k\), define
\[
    f(a)\triangleq
    \frac{\min\{a+r-1,s\}}{a}.
\]
Before the numerator reaches \(s\), we have \(f(a)=1+(r-1)/a\), and after it reaches \(s\), we have \(f(a)=s/a\). Hence \(f(a)\) is nonincreasing in \(a\), and the minimum over \(1\leq a\leq k\) is attained at \(a=k\). For \(a\geq k\), the denominator is fixed at \(k\), while the numerator is nondecreasing in \(a\), so the minimum over this range is also attained at \(a=k\). Consequently,
\[
    \overline C_{\rm MDS}(\Sigma) =\frac{\min\{r+k-1,s\}}{k},
\]
as claimed.
%For $a<k$, the ratio is at least its value at $a=k$, because the denominator is smaller.  For $a\ge k$, the denominator is $k$ and the numerator is minimized at $a=k$.  Hence, the minimum is achieved by any cyclic interval $A$ of length $k$, which gives \eqref{eq:cyclic-CMDS}.
\end{pf}

The converse has two regimes. If \(r\ge s-k+1\), its value is \(s/k\), and nonzero direct rows are available. If \(r\le s-k\), its value is \((r+k-1)/k\), while all local sections of the target space are zero. The constructions below match this division.

\subsection{Achievability results}
We first state the resulting bounds. The rest of the section gives the constructions and their proofs. 

\begin{thm}\label{thm:MDS_cyclic_lowerbound}
    Let \(T\in\F_q^{k\times s}\) be an MDS matrix, let \(m=s\), and consider the cyclic source-access system. Then the lower bound for the computing capacity is
    \begin{equation}
        C(\cN_{s,s,\Sigma},T)\ge \begin{cases}
            \frac{s}{k}, & \text{if $r\ge s-k+1$;}\\
            \lfloor\frac{k+r-1}{k}\rfloor, &\text{if $r<s-k+1$}.
        \end{cases}
    \end{equation}
    Hence, the converse is attained in the dense regime, while its floor is achieved in the sparse regime.
\end{thm}
%In the dense regime $r\ge s-k+1$, since $\mathscr{W}=\mathscr{C}_T^{\ell}$ already has sufficiently rich local sections, and no auxiliary direction is needed. In the sparse regime, however, for every $i\in[s]$, 
%\[\mathscr{C}_T^{\ell}\cap\mathscr{E}_{\Sigma_i}=\{\Zero\},\]
%so a direct target-row realization cannot transmit any nonzero target direction locally. Capacity achievement therefore requires a strict enlargement of the lifted target row space.

%For the divisible sparse regime, we construct such an enlargement indirectly through a sliding parity-check design. The construction first produces a family of decoder columns whose local windows span the target fibers and then derives a support-constrained received matrix. The representation theorem subsequently converts this matrix into the auxiliary-row normal form. This argument applies to arbitrary \(r<s-k+1\). For nondivisible parameters, we restrict each cyclic access set to a smaller divisible one and apply the same construction, which yields the integer-part rate in Theorem~\ref{thm:MDS_cyclic_lowerbound}.

%The remainder of this section proves the three parts in order. We first treat the dense regime using locally shortened MDS codewords. We then establish the sliding parity-check construction for the divisible sparse regime and finally derive the nondivisible result by support restriction.
In the dense regime, we use the smallest global space \(\mathscr C_T^{(\ell)}\), so no auxiliary rows are needed. In the divisible sparse regime, we build a support-constrained matrix \(E\) whose row space strictly contains \(\mathscr C_T^{(\ell)}\). The construction is based on overlapping parity-check windows. In the nondivisible sparse regime, we reduce the support size to the largest smaller value covered by the divisible construction.

\subsubsection{Dense regime $r\ge s-k+1$.}
For each $i\in[s]$, let $\Delta_i=[i:i+s-k]_{\rm s}.$ Since \(|\Delta_i|=s-k+1\), we have \(\Delta_i\subseteq\Sigma_i\).   Lemma~\ref{lem:mds-shortening-threshold} shows that the shortened space $\mathscr \mathscr{C}_T\cap \mathscr E_{\Delta_i}$ is one-dimensional. Choose a nonzero row vector \begin{equation}
    \bc_i\in\cC_T\cap\cE_{\Delta_i}^{(1)}.
    \label{eq:dense-ci}
\end{equation}
Since $T$ is MDS, every nonzero vector in $\mathscr{C}_T$ has Hamming weight at least $s-k+1$.  Thus, $\supp(\bc_i)=\Delta_i.$

Every \(k\) cyclically consecutive vectors
\begin{equation}
    \bc_a,\bc_{a+1},\ldots,\bc_{a+k-1}
\label{eq:dense-consecutive-basis}
\end{equation}
form a basis of \(\cC_T\). To see this, restrict them to the coordinates \([a:a+k-1]_{\rm s}\). For \(t=0,\ldots,k-1\), the vector \(\bc_{a+t}\) is zero in the first \(t\) coordinates and nonzero in coordinate \(a+t\). The restricted \(k\times k\) matrix is triangular with a nonzero diagonal.

We now use one such basis for each of \(s\) target instances. In instance \(a\), the nodes \(a,a+1,\ldots,a+k-1\) send the target rows in \eqref{eq:dense-consecutive-basis}. The sink recovers that instance because these rows form a basis of \(\cC_T\). Each middle node appears in exactly \(k\) of the \(s\) basis blocks. Thus, every middle node sends \(k\) symbols, giving a direct \((s,k)\) code of rate \(s/k\). This matches the converse in Lemma~\ref{lem:cyclic-bound}.

The main idea of the coding scheme is illustrated by the following example.

\begin{figure}
    \centering
    \includegraphics[width=0.5\linewidth]{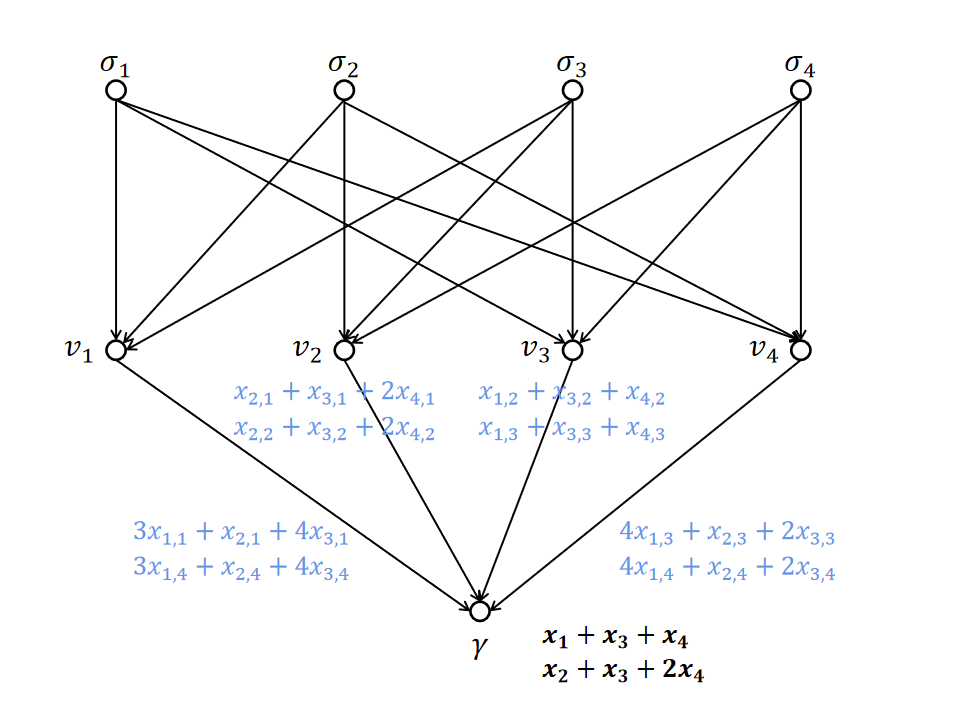}
    \caption{An example of the dense cyclic construction.}
    \label{fig:dense_cyclic}
\end{figure}

\begin{example}[Dense cyclic construction]\label{ex:dense-cyclic}
Let \(s=m=4\), \(k=2\), \(r=3\), and $\F_q=\F_5$.  
Take
\[
        T=\begin{bmatrix}
        1&0&1&1\\
        0&1&1&2
        \end{bmatrix},
\]
and 
\[
\Sigma_1=\{1,2,3\},\quad
\Sigma_2=\{2,3,4\},\quad
\Sigma_3=\{3,4,1\},\quad
\Sigma_4=\{4,1,2\}.
\]
The four vectors selected by \eqref{eq:dense-ci} can be chosen as
\[
\begin{aligned}
    \bc_1&=(3,1,4,0),&
    \bc_2&=(0,1,1,2),\\
    \bc_3&=(1,0,1,1),&
    \bc_4&=(1,4,0,4).
\end{aligned}
\]
Their supports are \(\Sigma_1,\ldots,\Sigma_4\), respectively, and every two consecutive vectors form a basis of \(\cC_T\). For instance \(a\), nodes \(a\) and \(a+1\) send the combinations defined by \(\bc_a\) and \(\bc_{a+1}\). Each node sends in two instances. This gives a direct \((4,2)\) code of rate \(2\) as shown in Fig.~\ref{fig:dense_cyclic}.
\end{example}

\subsubsection{Divisible sparse regime $r<s-k+1$ and $k\mid (r+k-1)$} Denote
\begin{equation}
    \ell=\frac{r+k-1}{k},
    \qquad
    N=k\ell=r+k-1.
    \label{eq:sparse-N}
\end{equation}
Since \(r\le s-k\), we have \(N<s\). The converse equals \(N/k=\ell\), so it is enough to construct an \((\ell,1)\) code.

From the support-constrained matrix representation defined in Section~\ref{sec:framework}, we need to construct a \(\Sigma\)-support-constrained received matrix \(E\in\mathbb F_q^{s\times s\ell}\) and a decoder $D\in\mathbb F_q^{N\times s}$ such that
\begin{equation}
    DE=I_\ell\otimes T.
    \label{eq:sparse-decoder-factorization}
\end{equation}
Let \(\bt_j=T(:,j)\). The \(N\times\ell\) block of target columns associated with source \(j\) spans the target fiber
\begin{equation}
    \mathscr U_j
    \triangleq
    \Span\{\be_b\otimes\bt_j:b\in[\ell]\}
    \subseteq\mathbb F_q^N,
    \label{eq:sparse-target-fiber}
\end{equation}
where $\be_b\in\F_q^{\ell}$ is the unit vector  with the $b$-th coordinate $1$. Eq.~\eqref{eq:sparse-decoder-factorization} implies that the vector $\be_a\otimes\bt_j$ is a linear combination of the columns of $D$. Denote 
\[D=\begin{bmatrix}
    \bd_1&\bd_2&\cdots&\bd_s
\end{bmatrix}.\]
Since $\Gamma(j)=[j-r+1:j]_{\rm s}$ is the set of middle nodes that observe source \(j\), the columns of \(E\) associated with source \(j\) can be supported only on the rows indexed by
\(\Gamma(j)\). Therefore, the required local condition is
\begin{equation}
    \mathscr U_j
    \subseteq
    \Span\{\bd_i:i\in\Gamma(j)\},
    \label{eq:sparse-local-decoding-goal}
\end{equation}
for every $j\in[s]$.
Indeed, once~\eqref{eq:sparse-local-decoding-goal} holds, the coefficients of these local linear representations give the entries of \(E\).

The parity-check viewpoint explains how we enforce all containments at once. Since $N-r=k-1,$ we seek, for each source \(j\), a rank-\((k-1)\) parity-check matrix \(H_j\). Its kernel has dimension \(r\). We will arrange that \(\mathscr U_j\subseteq\ker H_j\) and that the \(r\) decoder columns indexed by \(\Gamma(j)\) form a basis of this kernel. The main difficulty is that the same decoder column belongs to several such local bases. The sliding construction below makes these overlapping requirements compatible.

\begin{figure}
    \centering
    \includegraphics[width=0.5\linewidth]{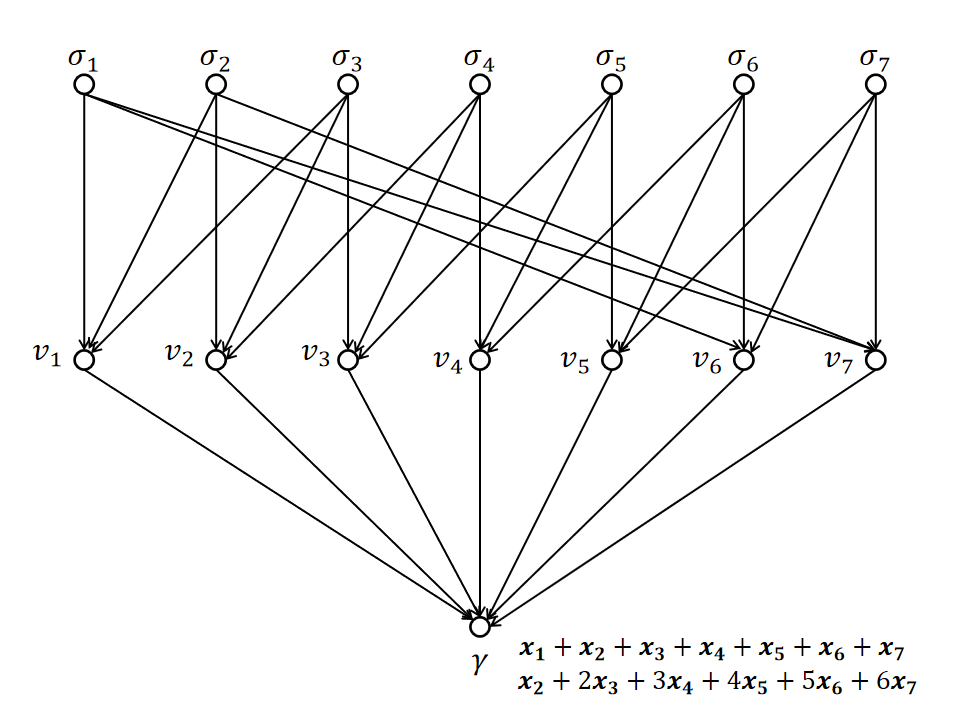}
    \caption{The network in Example~\ref{ex:sparse-parity-check}.}
    \label{fig:sparse_cyclic}
\end{figure}

\begin{example}\label{ex:sparse-parity-check}
Consider the cyclic network in Fig.~\ref{fig:sparse_cyclic}. Let $k=2, s=7$ and $r=3$, then $\ell=(r+k-1)/k=2$. 
Denote \(N=k\ell=4=r+k-1\). Assume that $T\in\F_7^{2\times 7}$ and
\[
    T=
    \begin{bmatrix}
        1&1&1&1&1&1&1\\
        0&1&2&3&4&5&6
    \end{bmatrix},
\]
then \(\bt_j=(1,j-1)^T\). 

For each \(i\), we first choose a parity-check direction \(\bq_i\in\F_7^2\) orthogonal to \(\bt_i\):
\[
\begin{aligned}
\bq_1&=(0,1)^T,&
\bq_2&=(6,1)^T,&
\bq_3&=(5,1)^T,&
\bq_4&=(4,1)^T,\\
\bq_5&=(3,1)^T,&
\bq_6&=(2,1)^T,&
\bq_7&=(1,1)^T.
\end{aligned}
\]

Next, we lift each direction to the two-dimensional space
\[
    \mathcal H_i=\F_7^2\otimes\Span(\bq_i)\subseteq\F_7^4.
\]
Moreover, we select $\bh_i,i\in[7]$ such that  every four cyclically consecutive \(\bh_i\)'s form a basis of \(\F_7^4\). For instance, we choose
\[
\begin{aligned}
\bh_1&=(0,0,0,1)^T,&
\bh_2&=(0,0,6,1)^T,&
\bh_3&=(5,1,0,0)^T,\\
\bh_4&=(4,1,0,0)^T,&
\bh_5&=(0,0,3,1)^T,&
\bh_6&=(2,1,2,1)^T,\\
\bh_7&=(1,1,0,0)^T.
\end{aligned}
\]

To determine $\bd_i,i\in[7]$, consider the $i$-th middle node. Since $i\in\Gamma(i),\Gamma(i+1)$ and $\Gamma(i+2)$, \(\bd_i\ne0\) should be orthogonal to \(\bh_i,\bh_{i+1},\bh_{i+2}\). One possible choice is
\[
\begin{aligned}
\bd_1&=(4,1,0,0)^T,&
\bd_2&=(0,0,1,1)^T,&
\bd_3&=(0,0,2,1)^T,\\
\bd_4&=(6,4,2,1)^T,&
\bd_5&=(2,5,2,1)^T,&
\bd_6&=(5,2,1,0)^T,\\
\bd_7&=(6,1,0,0)^T.
\end{aligned}
\]
Every three cyclically consecutive \(\bd_i\)'s are linearly independent.

Now, we should check \eqref{eq:sparse-local-decoding-goal} is satisfied for every $j\in[7]$.
Take $j=1$ for example, source \(1\) is observed by the nodes \(\Gamma(1)=\{6,7,1\}\). Its parity-check matrix has the single row \(H_1=\bh_1^T\), and hence
\[
    \ker H_1
    =\{(x_1,x_2,x_3,x_4)^T:x_4=0\}.
\]
The three vectors \(\bd_6,\bd_7,\bd_1\) belong to this nullspace and are independent, so they form a basis of \(\ker H_1\). Thus, the target fiber
\[
    \mathscr U_1
    =\Span\{(1,0,0,0)^T,(0,0,1,0)^T\}
    \subseteq\ker H_1=\Span\{\bd_6,\bd_7,\bd_1\}.
\]
More explicitly,
\[
    (1,0,0,0)^T=4\bd_7+3\bd_1,
    \qquad
    (0,0,1,0)^T=\bd_6+5\bd_7.
\]
The same calculation works for all sources. Then, we obtain
\[D=\begin{bmatrix}
    4&0&0&6&2&5&6\\1&0&0&4&5&2&1\\0&1&2&2&2&1&0\\0&1&1&1&1&0&0
\end{bmatrix},\qquad E=\addtocounter{MaxMatrixCols}{10}\begin{bmatrix}
    3&6&2&0&0&0&0&0&0&0&0&0&0&0\\
    0&0&0&0&0&0&0&0&1&3&5&0&0&0\\
    0&0&0&1&4&0&0&0&0&6&5&4&0&0\\
    0&0&0&6&4&2&0&0&0&0&0&0&0&0\\
    0&0&0&0&6&5&0&0&0&0&0&0&5&6\\
    0&0&0&0&0&0&0&1&0&0&0&0&5&3\\
    4&2&0&0&0&0&6&5&0&0&0&0&0&6
\end{bmatrix}.\]
The matrix $E$ is $\Sigma$-support-constrained and $DE=I_2\otimes T$. Thus, the seven middle nodes each transmit one symbol and the sink recovers two
instances, giving rate two, equal to the converse.
\end{example}

We now give the general construction. The first lemma extracts a cyclic family of parity-check directions from the MDS target.

\begin{lem}\label{lem:consecutive-mds-parity-directions}
Let \(T\in\F_q^{k\times s}\) be MDS and write \(\bt_j=T(:,j)\). There exist \(\bq_i\in\mathbb F_q^k\setminus\{\Zero\}\) such that for every $j\in[i-k+2:i]_{\rm s}$,
\begin{equation}
    \bq_i^T\bt_j=0.
    \label{eq:sparse-h-definition}
\end{equation}
When \(k=1\), the interval is empty and any nonzero \(\bq_i\) may be used. Moreover, every \(k\) cyclically consecutive vectors among the \(\bq_i\)'s are linearly independent.
\end{lem}

\begin{pf}
For \(k\ge2\), the \(k-1\) columns $\bt_j$ indexed by \([i-k+2:i]_{\rm s}\) are linearly independent. Their left nullspace is one-dimensional, so \(\bq_i\) exists and is unique up to a nonzero scalar.

Fix a cyclic starting index \(a\). Pair the rows
\(\bq_a^T,\ldots,\bq_{a+k-1}^T\) with the columns
\(\bt_{a+1},\ldots,\bt_{a+k}\). If \(0\le u<t\le k-1\), then \(\bq_{a+t}^T\bt_{a+u+1}=0\) by~\eqref{eq:sparse-h-definition}. Moreover \(\bq_{a+t}^T\bt_{a+t+1}\ne0\); otherwise \(\bq_{a+t}\) would be orthogonal to \(k\) columns of the MDS matrix \(T\), which are linearly independent. The resulting pairing matrix is triangular with nonzero diagonal. Hence the \(k\) consecutive vectors are independent.
\end{pf}

For each \(i\in[s]\), define the \(\ell\)-dimensional lifted parity-check space
\begin{equation}
    \mathcal H_i
    \triangleq
    \mathbb F_q^\ell\otimes\Span(\bq_i)
    \subseteq\mathbb F_q^N.
    \label{eq:sparse-H-space}
\end{equation}
The next lemma selects one vector from each lifted space while making every sliding window nonsingular.

\begin{lem}\label{lem:simultaneous-sliding-representatives}
Assume \(N=k\ell<s\), and let the spaces \(\mathcal H_i\) be defined by \eqref{eq:sparse-H-space}. If \(q>N\), then there exist vectors \(\bh_i\in\mathcal H_i\), \(i\in[s]\), such that every \(N\) cyclically consecutive vectors $\bh_a,\bh_{a+1},\ldots,\bh_{a+N-1}$ form a basis of \(\mathbb F_q^N\).
\end{lem}
\begin{pf}
We first prove that for every cyclic interval $[a:a+N-1]_{\rm s}$, there exist $N$ linearly independent vectors $\bh_a,\bh_{a+1},\ldots,\bh_{a+N-1}$ such that $\bh_i\in\cH_i$, and then show that $\{\bh_i:i\in[s]\}$ can be chosen for all cyclic windows at the same time.

The first statement is proved by leveraging Lemma~\ref{lem:linear-rado}. For a fixed cyclic interval \(I=[a:a+N-1]_{\rm s}\) and any subset \(J\subseteq I\), it suffices to show 
\[\dim\left(\sum_{i\in J}\cH_i\right)\ge |J|.\]
Let $d_J\eqdef\dim\Span\{\bq_i:i\in J\}$, and partition \(I\) into \(\ell\) consecutive blocks \(I_1,\ldots,I_\ell\), each of length \(k\). From Lemma~\ref{lem:consecutive-mds-parity-directions}, the vectors \(\bq_i\) in each block are independent. Hence, for every $b\in[\ell]$,
\[d_J\ge |J\cap I_b|.\]
It follows that
\begin{equation}
    \dim\left(\sum_{i\in J}\mathcal H_i\right)=\ell d_J\ge|J|.
    \label{eq:sparse-rado-condition}
\end{equation}

It remains to make the choices work for all cyclic windows at the same time. Since $\bh_i\in\cH_i$, we can write
\[\bh_i=\sum_{b=1}^{\ell}z_{i,b}(\be_b\otimes\bq_i),\]
where the \(z_{i,b}\)'s are variables. Let \(\Delta_a\) be the determinant of the window starting at \(a\). The fixed-window argument shows that every \(\Delta_a\) is a nonzero polynomial. Therefore
\[
    \Delta\triangleq\prod_{a=1}^{s}\Delta_a
\]
is also nonzero. Each variable \(z_{i,b}\) occurs in exactly \(N\) cyclic windows and occurs with degree at most one in each determinant. Its individual degree in \(\Delta\) is therefore at most \(N\).

A nonzero polynomial over \(\mathbb F_q\) whose degree in each individual variable is smaller than \(q\) cannot vanish at every point of \(\mathbb F_q^{s\ell}\). This follows by induction on the number of variables from the usual univariate root bound. Since \(q>N\), there is an evaluation for which \(\Delta\ne0\). At this evaluation every cyclic length-\(N\) window is a basis.
\end{pf}

We next turn the sliding parity-check sequence into decoder columns.

\begin{lem}\label{lem:sliding-parity-check-realization}
Let \(N=k\ell=r+k-1\), and let \(\bh_1,\ldots,\bh_s\) satisfy
Lemma~\ref{lem:simultaneous-sliding-representatives}. For every \(i\in[s]\),
choose
\begin{equation}
    \bd_i\in
    \Span\{\bh_i,\bh_{i+1},\ldots,\bh_{i+N-2}\}^{\perp}
    \setminus\{\Zero\}.
    \label{eq:sparse-p-definition}
\end{equation}
Then every \(r\) cyclically consecutive \(\bd_i\)'s are linearly independent
and, for every source \(j\),
\begin{equation}
    \mathscr U_j\subseteq\Span\{\bd_i:i\in\Gamma(j)\}.
    \label{eq:sparse-fiber-containment}
\end{equation}
\end{lem}

\begin{pf}
From Lemma~\ref{lem:simultaneous-sliding-representatives}, the \(N-1\) vectors in~\eqref{eq:sparse-p-definition} are independent, so their orthogonal complement is one-dimensional and \(\bd_i\) exists.
To prove the first claim, fix \(a\in[s]\) and pair
\(\bd_a,\ldots,\bd_{a+r-1}\) with
\(\bh_{a-1},\ldots,\bh_{a+r-2}\). For \(u>t\), the definition of
\(\bd_{a+t}\) gives
\[
    \bd_{a+t}^T\bh_{a-1+u}=0.
\]
The diagonal term \(\bd_{a+t}^T\bh_{a+t-1}\) is nonzero. Otherwise, \(\bd_{a+t}\) would be orthogonal to the \(N\)-vector basis 
\[
    \bh_{a+t-1},\bh_{a+t},\ldots,\bh_{a+t+N-2}.
\]
The pairing matrix is triangular with nonzero diagonal. Hence the \(r\) consecutive vectors \(\bd_a,\ldots,\bd_{a+r-1}\) are independent.

For source \(j\), define parity-check matrix $H_j\in\F_q^{(k-1)\times N}$
\begin{equation}
    H_j\triangleq
    \begin{bmatrix}
        \bh_j^T\\
        \bh_{j+1}^T\\
        \vdots\\
        \bh_{j+k-2}^T
    \end{bmatrix}.
    \label{eq:sparse-Hj}
\end{equation}
Then, \(\dim(\ker H_j)=N-(k-1)=r\). 
From \eqref{eq:sparse-p-definition}, the vectors $\bd_i,i\in[j-r+1:j]_{\rm s}=\Gamma(j)$ satisfy $H_j\bd_i=\Zero$. Combining with the independence of $\{\bd_i:i\in\Gamma(j)\}$, these vectors form a basis of $\ker H_j$.

Moreover, for every \(i\in[j:j+k-2]_{\rm s}\), the interval in \eqref{eq:sparse-h-definition} contains \(j\). Therefore \(\bq_i^T\bt_j=0\), and since \(\bh_i\in\mathcal H_i\), we have for every $b\in[\ell]$,
\[\bh_i^T(\be_b\otimes\bt_j)=0.\]
It follows that
\begin{equation}
    \mathscr{U}_j\subseteq\ker H_j=\Span(\{\bd_i:i\in\Gamma(j)\}).
    \label{eq:sparse-U-in-kernel}
\end{equation}
\end{pf}
We can now complete the support-constrained realization.

\begin{thm}\label{thm:divisible-sparse-cyclic-mds}
Let \(T\in\F_q^{k\times s}\) be an MDS matrix and let \(\Sigma_i=[i:i+r-1]_{\rm s}\). Suppose $r\le s-k$ and $r+k-1=k\ell.$ Let \(t\) be any positive integer such that $\bar q\triangleq q^t>k\ell.$ Then there exists an \(\F_{\bar q}\)-linear \((\ell,1)\) computation code. Equivalently, there exists an \(\F_q\)-linear \((t\ell,t)\) computation code.
Consequently, the computing capacity can be achieved by linear network codes, i.e.,
\[
    C(\mathcal N_{s,s,\Sigma},T)
    =C_{\rm lin}(\mathcal N_{s,s,\Sigma},T)
    =\ell.
\]
 
\end{thm}

\begin{pf}
Let \(N=k\ell\). We first work over the extension field \(\F_{\bar q}\). Since all \(k\times k\) minors of \(T\) that are nonzero over \(\F_q\) remain nonzero over \(\F_{\bar q}\), the matrix \(T\) is also MDS over \(\F_{\bar q}\). Moreover, \(\bar q>N\). Hence, Lemmas~\ref{lem:consecutive-mds-parity-directions}--\ref{lem:sliding-parity-check-realization} can be applied over \(\F_{\bar q}\) to select the vectors $\bh_i,\bq_i$ and $\bd_i$. Define
\[
    D=[\bd_1\ \bd_2\ \cdots\ \bd_s]\in\mathbb F_{\bar q}^{N\times s}.
\]
For every source \(j\) and instance \(b\), \eqref{eq:sparse-fiber-containment} gives coefficients \(c_{i,j,b}\in\mathbb F_{\bar q}\), \(i\in\Gamma(j)\), such that
\begin{equation}
    \be_b\otimes\bt_j
    =\sum_{i\in\Gamma(j)}c_{i,j,b}\bd_i.
    \label{eq:sparse-local-coefficients}
\end{equation}
Set \(c_{i,j,b}=0\) whenever \(i\notin\Gamma(j)\), and define \(E\in\mathbb F_{\bar q}^{s\times s\ell}\) by $E_{i,(b-1)s+j}=c_{i,j,b}.$
The entry in row \(i\) and source block \(j\) can be nonzero only when \(i\in\Gamma(j)\), equivalently when \(j\in\Sigma_i\). Therefore \(E\) is \(\Sigma\)-support-constrained. Moreover, \eqref{eq:sparse-local-coefficients} gives
\begin{equation}
    DE=I_\ell\otimes T.
    \label{eq:sparse-PFbar}
\end{equation}
Thus the middle nodes transmit \(E\bx_S\), one symbol each, and the sink applies \(D\) to recover the \(\ell\) target instances.

It remains to express this code over the original field \(\F_q\). Fix an \(\F_q\)-basis of \(\F_{\bar q}\), where \(\bar q=q^t\), and expand every \(\F_{\bar q}\)-symbol into its \(t\) coordinates over \(\F_q\). Every \(\F_{\bar q}\)-linear operation is then an \(\F_q\)-linear operation. Since the entries of \(T\) belong to \(\F_q\), the target transformation acts component-wise on these \(t\) coordinates. Consequently, the extension field \((\ell,1)\) code induces an \(\F_q\)-linear \((t\ell,t)\) code with the same rate.

Finally, Lemma~\ref{lem:cyclic-bound} gives
\[
    C(\mathcal N_{s,s,\Sigma},T)\leq\frac{r+k-1}{k}=\ell.
\]
The constructed linear code achieves this upper bound, and hence
\[
    C(\mathcal N_{s,s,\Sigma},T)=C_{\rm lin}(\mathcal N_{s,s,\Sigma},T)=\ell.
\]
\end{pf}

This construction fits the normalized form in Corollary~\ref{cor:normalized-row-space-realization}. Indeed, $I_\ell\otimes T=DE$ implies that the target space is contained in \(\Row(E)\). Hence there exist \(V\) and an invertible \(P\) such that 
\[
    E=P
    \begin{bmatrix}
        I_\ell\otimes T\\V
    \end{bmatrix}.
\]
The sliding argument finds this realization in the reverse order: it first designs the decoder columns, then obtains the supported rows of \(E\) from the local fiber representations. The framework then identifies the corresponding auxiliary rows \(V\) and row transformation \(P\).

\begin{rmk}
    The sliding construction requires a coefficient field with more than \(N=r+k-1\) elements. This causes no restriction on the original base field. Indeed, for any \(T\) over \(\F_q\), choose \(t\) such that \(q^t>N\), apply the construction over \(\F_{q^t}\), and expand every extension-field symbol into \(t\) symbols over \(\F_q\). Since the entries of \(T\) lie in \(\F_q\), this gives an \(\F_q\)-linear \((t\ell,t)\) code with the same rate. We may therefore present the construction over a field of size greater than \(N\).
\end{rmk}

\subsubsection{Nondivisible sparse regime $r<s-k+1$ and $k\nmid k+r-1$}
Define 
\begin{equation}
    \ell_0
    \triangleq
    \left\lfloor\frac{r+k-1}{k}\right\rfloor,
    \qquad
    r_0
    \triangleq
    k\ell_0-k+1.
    \label{eq:nondivisible-r0}
\end{equation}
Then $r_0+k-1=k\ell_0$ and $ r_0\le r.$
Thus \(r_0\) is the largest support size not exceeding \(r\) for which the sparse converse is an integer. The key point is that a code for cyclic supports of size \(r_0\) is also a code for supports of size \(r\), since the extra source access can simply be ignored.

\begin{cor}\label{cor:cyclic-integer-part}
Let \(T\in\F_q^{k\times s}\) be MDS and let
\(\Sigma_i=[i:i+r-1]_{\rm s}\), where \(r\le s-k\). Then
\begin{equation}
    \left\lfloor\frac{r+k-1}{k}\right\rfloor
    \le
    C_{\rm lin}(\mathcal N_{s,s,\Sigma},T)
    \le
    C(\mathcal N_{s,s,\Sigma},T)
    \le
    \frac{r+k-1}{k}.
    \label{eq:nondivisible-capacity-bounds}
\end{equation}
In particular, the gap between the achievable rate and the upper bound is strictly smaller than one. 
\end{cor}

\begin{pf}
For every \(i\in[s]\), restrict the original access set to
\[
    \Sigma_i^{(0)}
    \triangleq
    [i:i+r_0-1]_{\rm s}
    \subseteq\Sigma_i.
\]
Because \(r_0\le r\le s-k\), the restricted system is in the sparse regime, and $r_0+k-1=k\ell_0.$
Theorem~\ref{thm:divisible-sparse-cyclic-mds} therefore gives a code of rate \(\ell_0\) for the restricted supports. Every transmitted row supported on \(\Sigma_i^{(0)}\) is also supported on \(\Sigma_i\), so the same code is valid for the original network. This proves the lower bound in \eqref{eq:nondivisible-capacity-bounds}. The upper bound follows from Lemma~\ref{lem:cyclic-bound}.

Finally,
\[
    0\le\frac{r+k-1}{k}-\ell_0<1,
\]
which proves the gap statement.
\end{pf}

The dense construction, Theorem~\ref{thm:divisible-sparse-cyclic-mds}, and Corollary~\ref{cor:cyclic-integer-part} together prove Theorem~\ref{thm:MDS_cyclic_lowerbound}.

\section{Conclusion}
\label{sec:conclusion}
This paper studied vector-linear function computation over three-layer source-access networks with a prescribed target function and a fixed network topology. We developed a support-constrained row-space framework in which the global row space received by the sink is the main design object. For any candidate row space containing the lifted target space, its minimum communication load is determined exactly by the ranks of its locally supported sections. Optimizing over this row space gives a variational characterization of the linear computing capacity and unifies direct target-row coding with auxiliary-row coding. For cyclic networks with MDS targets, we further showed how parity-check-based enlargements of the target space overcome local support-constraints, attain the converse bound in the dense regime and in the divisible sparse regime, and provide a floor-rate guarantee otherwise.

The framework separates the coding problem into two parts. Once the global row space is fixed, its realization under the prescribed access sets is completely characterized by local rank inequalities. The remaining difficulty is therefore concentrated in choosing a suitable global row space. This viewpoint preserves the algebraic structure of the target matrix and makes clear that computing capacity depends on the compatibility between the target and the access pattern, rather than only on the target rank or the network topology. It also distinguishes the fixed-function, fixed-topology problem considered here from linearly separable computation settings in which the demand, the task assignment, or both may be designing variables.

An important open problem is to develop systematic methods for selecting a capacity-optimal enlarged row space for general target matrices and access systems. For cyclic MDS instances, closing the fractional gap in the nondivisible sparse regime appears to require a genuinely block-valued construction and a better understanding of the rank growth of the associated sliding parity-check spaces. Other directions include identifying target classes beyond MDS and sum matrices for which the outer row-space optimization is tractable, reducing the field-size and subpacketization requirements of the constructions, and extending the local-section approach to more general network topologies.

\bibliographystyle{IEEEtranS}
\bibliography{reference}

\end{document}